\documentclass[lettersize,journal]{IEEEtran}

\usepackage{color}
\usepackage{amsmath,amsthm,amssymb,amsfonts}

\allowdisplaybreaks[4]

\usepackage[T1]{fontenc}
\usepackage{graphicx}
\usepackage{longtable}
\usepackage{float}
\usepackage{hyperref}

\usepackage{bm}

\usepackage[bottom]{footmisc}

\usepackage[mathscr]{eucal}

\usepackage[caption=false,font=normalsize,labelfont=sf,textfont=sf]{subfig}

\usepackage{multirow}
\usepackage{makecell}

\usepackage{cite}

\usepackage{epsfig}

\usepackage{soul}

\usepackage[lined,boxed,ruled]{algorithm2e}

\newtheorem{theorem}{Theorem}
\newtheorem{lemma}{Lemma}

\newtheorem{corol}{Corollary}

\newcommand{\Exp}{{\mathbb{E}}}

\newcommand{\braces}[1]{\left\lbrace #1\right\rbrace}

\newcommand{\Dkl}[2]{D_{\textrm{KL}}\left\lbrace #1 : #2\right\rbrace}

\newcommand{\setposi}[1]{\mathcal{Z}_{#1}^+}

\newcommand{\setnnega}[1]{\mathcal{Z}_{#1}}

\newcommand{\diag}[1]{\mathrm{diag}\left\lbrace #1\right\rbrace }
\newcommand{\Diag}[1]{\mathrm{Diag}\left\lbrace #1\right\rbrace }
\newcommand{\Bdiag}[1]{\mathrm{Bdiag}\left\lbrace #1\right\rbrace }

\newcommand{\argmax}[1]{\mathop{\arg\max}\limits_{#1}}
\newcommand{\argmin}[1]{\mathop{\arg\min}\limits_{#1}}

\newcommand{\equaa}{\mathop{=}^{(\textrm{a})}}
\newcommand{\equab}{\mathop{=}^{(\textrm{b})}}

\newcommand{\ba}{\mathbf{a}}

\newcommand{\bd}{\mathbf{d}}
\newcommand{\be}{\mathbf{e}}

\newcommand{\bs}{\mathbf{s}}
\newcommand{\bt}{\mathbf{t}}

\newcommand{\bx}{\mathbf{x}}
\newcommand{\by}{\mathbf{y}}
\newcommand{\bz}{\mathbf{z}}

\newcommand{\bA}{\mathbf{A}}

\newcommand{\bC}{\mathbf{C}}

\newcommand{\bG}{\mathbf{G}}

\newcommand{\bI}{\mathbf{I}}

\newcommand{\bR}{\mathbf{R}}

\newcommand{\bX}{\mathbf{X}}

\newcommand{\bbC}{\mathbb{C}}
\newcommand{\bbR}{\mathbb{R}}

\newcommand{\bSigma}{{\boldsymbol\Sigma}}

\newcommand{\btheta}{{\boldsymbol\theta}}

\newcommand{\bmu}{{\boldsymbol\mu}}

\newcommand{\bxi}{{\boldsymbol\xi}}

\newcommand{\bbeta}{{\boldsymbol\eta}}

\newcommand{\normmm}[1]{{\left\vert\kern-0.25ex\left\vert\kern-0.25ex\left\vert #1 
		\right\vert\kern-0.25ex\right\vert\kern-0.25ex\right\vert}}

\definecolor{myback}{RGB}{204,232,207}

\usepackage{makecell}

\begin{document}

\title{Signal Detection for Ultra-Massive MIMO: An Information Geometry Approach}	
\author{
	Jiyuan~Yang,~\IEEEmembership{Student~Member,~IEEE,}
	~Yan~Chen,
    ~Xiqi~Gao,~\IEEEmembership{Fellow,~IEEE,}
    ~Dirk~Slock,~\IEEEmembership{Fellow,~IEEE,}
    and~Xiang-Gen~Xia,~\IEEEmembership{Fellow,~IEEE}
}

\maketitle
\begin{abstract}
In this paper, we propose an information geometry approach (IGA) for signal detection (SD) in ultra-massive  multiple-input multiple-output (MIMO) systems.
We formulate the signal detection as obtaining the marginals of the {\textsl{a posteriori}} probability distribution of the transmitted symbol vector.
Then, a maximization of the {\textsl{a posteriori}} marginals (MPM) for signal detection can be performed.
With the information geometry theory, we calculate the approximations of the \textsl{a posteriori} marginals.
It is formulated as an iterative $\mathbf{\textsl{m}}$-projection process between submanifolds with different constraints.
We then apply the central-limit-theorem (CLT) to simplify the calculation of the $\mathbf{\textsl{m}}$-projection since the direct calculation of the $\mathbf{\textsl{m}}$-projection is of exponential-complexity.
With the CLT, we obtain an approximate solution of the $\mathbf{\textsl{m}}$-projection, 
which is asymptotically accurate.
Simulation results demonstrate that the proposed IGA-SD emerges as a promising and efficient method to implement the signal detector in ultra-massive MIMO systems.
\end{abstract}

\section{Introduction}
As one of the critical technologies for $5$G, massive multiple-input multiple-output (MIMO) can provide significant gains in both spectral efficiency and energy efficiency for communication systems \cite{8626085,10121037}.
In future $6$G communications, an ultra-massive MIMO system will employ an ultra-large array with hundreds or thousands of antennas,  serving tens or even hundreds of users simultaneously,
which is able to achieve higher spectral efficiency and energy efficiency, and wider and more flexible network coverage than ever \cite{10054381,bjornson2019massive,9170651,9693928}.
For the realization of the substantial benefits of ultra-massive MIMO, signal detection is of great importance.
Based on a received signal, the task of the detector is to determine the transmitted symbol.
The optimal detector based on the maximum  {\textsl{a posteriori}} (MAP) criterion or the maximum-likelihood (ML) criterion performs an exhaustive search and examines all possible symbols, which is shown as non-deterministic polynomial-time hard (NP-hard).
Consequently, the computational complexity of the MAP or ML detector  rapidly becomes unaffordable as the number of decision symbols increases.
On the other hand, the linear detectors, e.g., the linear minimum-mean-squared error (LMMSE) detector, are widely adopted due to the polynomial-time complexity.
Nonetheless, the  estimates of the transmitted symbols of the LMMSE detector are biased \cite{7244171}, and the performance of the LMMSE detector degrades severely in massive MIMO systems with high-order constellations \cite{6841617}.

In the past few decades, many works have been devoted to the massive MIMO signal detection \cite{6841617,9250659,6415398,9484686,9139393,5961820},
of which Bayesian inference approaches, e.g., belief propagation (BP), expectation propagation (EP), etc.,  are of significant interest due to the relatively low computational complexity and higher performance than linear detection.
These methods aim to calculate an approximation of the {\textsl{a posteriori}} probability distribution (or its marginals) of the transmitted symbols.
A hard-decision based on the \textsl{a posteriori} mean or a soft-decision based on the \textsl{a posteriori}  marginals is then performed. 
In \cite{6841617}, EP is first introduced into the massive MIMO signal detection with high-order modulation.
\cite{9484686} proposes a beam domain detector based on the layered BP for massive MIMO systems.
A variant EP detector is proposed in \cite{9139393} based on decentralized processing.
\cite{5961820} proposes a MIMO detector for high-order QAM modulation based on the Gaussian tree approximation.

Information geometry, which is introduced by Rao\cite{rao}, and then formally developed by Amari\cite{amari} and Cencov\cite{cencov}, has found a wide range of applications.
For Bayesian inference, the space defined by the parameters of the \textsl{a posteriori}
probability distribution is regarded as a differentiable manifold with a
Riemannian structure, and the definitions and tools of differential geometry are well applied by Amari et al. \cite{IGanditsAPP,srbpig}.
Amari et al. also show the intrinsic geometric insight of some classical Bayesian inference methods, e.g, the belief propagation (BP) \cite{BP}.
Meanwhile, some optimization methods, such as the concave-convex procedure (CCCP) \cite{CCCP}, are also applied to calculate the marginals of the \textsl{a posteriori} distribution.
On Bayesian inference in communications, \cite{1302292} analyzes the turbo and low-density parity-check (LDPC) codes from the perspective of information geometry, and an improvement of turbo and LDPC codes is  proposed from the geometrical view.
The information geometry is extended to complex signal processing and an information geometry approach is proposed  for massive MIMO channel estimation in \cite{IGA,SIGA}.

In addition to the unique insight that the geometric perspective offers, information geometry also provides us with a unified framework where various sets of probability distributions are considered to be endowed with the structure of differential geometry. 
Hence, we are able to construct a Fisher information matrix (FIM) based distance between two parametrized distributions.
Amari \cite{amari} also shows that this distance is invariant to non-singular transformation of the parameters.
As a result, information geometry is closely related to estimation theory.
Due to these advantages, information geometry has recently been applied to many other problems such as the complex network construction \cite{8066330}, the target detection \cite{9479799}, and the clustering \cite{9123433}.
 
In this paper, we propose an information geometry approach for signal detection (IGA-SD) for ultra-massive MIMO systems.
We formulate the signal detection as obtaining the marginls of the \textsl{a posteriori} probability distribution of the transmitted symbols.
Then, a component-wise decision can be performed based on the \textsl{a posteriori} marginals.
With the information geometry theory, we calculate the approximations of the \textsl{a posteriori} marginals.
More precisely, by treating the sets of the probability distributions of discrete random vectors with different constraints as several different (sub)manifolds, the calculation of the marginals is converted into an iterative $m$-projection process.
Furthermore, since the calculation of the $m$-projection in signal detection is of exponential-complexity, we apply the central-limit-theorem (CLT) to simplify its calculation.
With the CLT, we are able to find an  approximate solution of the $m$-projection, which is asymptotically accurate.  
At last, a soft-decision is performed based on the approximation of the \textsl{a posteriori} marginls.

The rest of the paper proceeds as follows. The system configuration and problem statement  are presented in Section \uppercase\expandafter{\romannumeral2}. 
Preliminaries of information geometry is introduced in Section \uppercase\expandafter{\romannumeral3}. 
The Information geometry approach for ultra-massive MIMO signal detection  is proposed in Section \uppercase\expandafter{\romannumeral4}. 
Simulation results are provided in Section \uppercase\expandafter{\romannumeral5}. 
The conclusion is drawn in Section \uppercase\expandafter{\romannumeral6}.

Notations: The following notations are adopted in this paper. 
Upper (lower) case boldface letters denote matrices
(column vectors). 
$\mathcal{R}\left(\cdot\right)$ and $\mathcal{I}\left(\cdot\right)$ denote the real and imaginary parts of a complex number (matrix), respectively.
The superscripts $\left( \cdot \right)^*$, $\left( \cdot \right)^T$ and $\left( \cdot \right)^H$ denote the conjugate, transpose and conjugate-transpose operator, respectively. 
$\Diag{\bx}$ denotes the diagonal matrix with $\bx$ along its main diagonal and $\diag{\bX}$ denotes a vector consisting of the diagonal elements of $\bX$.
$\Bdiag{\bX_1,\bX_2,\ldots}$ denotes a block diagonal matrix with matrices $\bX_i$ located along the main diagonal.
We use $a_{i,j}$ to denote the $\left(i,j\right)$-th element of the matrix $\bA$, where the element indices start with $1$. 
$\odot$ and $\otimes$ denote the Hadamard product and  Kronecker product, respectively. 
Define $\setnnega{N} \triangleq \braces{0,1,\ldots,N}$ and $\setposi{N} \triangleq \braces{1,2,\ldots,N}$.
$\setminus$ denotes the set subtraction operation.
To avoid confusion,  $p\left(\cdot\right)$ and $f\left(\cdot\right)$ denote the probability distribution of discrete random variables and the probability density function (PDF) of continuous random variables, respectively.
$f_{\textrm{CG}}\left( \bx;\bmu,\bSigma \right)$ denotes the PDF of a complex Gaussian distribution $\mathcal{CN}\left( \bmu,\bSigma \right)$ for vector $\bx$ of complex random variables.
$f_{\textrm{G}}\left( \bx;\bmu,\bSigma \right)$ denotes the PDF of a real Gaussian distribution $\mathcal{N}\left( \bmu,\bSigma \right)$ for vector $\bx$ of complex random variables.

\section{System Model and Problem Statement}

\subsection{System Configuration}
Consider an uplink ultra-massive MIMO system where one base station (BS) equipped with an ultra-massive antenna array serves $K$ single-antenna users, and the BS has $N_r$ antennas.
Denote the transmitted symbol vector of all users as $\tilde{\bs} \triangleq \left[\tilde{s}_1,\tilde{s}_2,\ldots,\tilde{s}_K  \right]^T \in\tilde{\mathbb{S}}^K$ where $\tilde{s}_k \in \tilde{\mathbb{S}}, k\in \setposi{K}$, is the  transmitted symbol of the user $k$. 
$\tilde{\mathbb{S}}$ is the signal constellation and let us assume 
$\tilde{\mathbb{S}} = \braces{\tilde{s}^{\left(0\right)},\tilde{s}^{\left(1\right)},\ldots,\tilde{s}^{\left(\tilde{L}-1\right)}}$, 
where $\braces{\tilde{s}^{\left(\ell\right)}}_{\ell=0}^{\tilde{L}-1}$ are the  constellation points, and	$\tilde{L}$ is the modulation order (or constellation size).
In this paper, we focus on the uncoded systems and the symmetric $\tilde{L}$-QAM modulation. We assume that each user chooses  symbols
from $\tilde{\mathbb{S}}$ uniformly at random, and all users use the same alphabet, although the proposed IGA-SD can be readily extended to arbitrary modulations with different distributions as long as the transmitted symbols of the users are statistically independent and the real and imaginary parts of each user's symbol are statistically independent as well.
We also assume that the average power of $\tilde{s}_k$ is normalized to unit, i.e., $\Exp \braces{ \left|\tilde{s}_k\right|^2}  = 1$, $k\in \setposi{K}$.
The symbol vector $\tilde{\bs}$ is then transmitted over a flat-fading complex channel,
and the received signal $\tilde{\by}\in\bbC^{N_r}$ at the BS can be modeled as
\begin{equation}\label{equ:rece model 1}
	\tilde{\by} = \tilde{\bG}\tilde{\bs} + \tilde{\bz},
\end{equation}
where $\tilde{\bG} \in\bbC^{N_r\times K}$ is the channel matrix,  
$\tilde{\bz}$ is an additive white circular-symmetric complex Gaussian noise vector, 
$\tilde{\bz} \sim\mathcal{CN}\left(\mathbf{0},\tilde{\sigma}_z^2\bI\right)$ 
and $\tilde{\sigma}_z^2$ is the noise variance.
In this work, we assume that the BS has perfect channel state information.
As a note, the reason why,  in the above notations, the tildes are added on the tops of the math symbols is that we will later formulate and analyze their real counterparts without the tildes for notational simplicity.

\subsection{Problem Statement}\label{sec:Probelm Statement}
Assuming that the transmitted symbol vector $\tilde{\bs}$ and the noise vector $\tilde{\bz}$ are independent with each other and the symbols transmitted by different users are independent with each other as well. Then, with the received signal model \eqref{equ:rece model 1}, the \textsl{a posteriori} probability distribution of the transmitted symbol vector $\tilde{\bs}$ can be expressed as 
\begin{equation}
	\begin{split}
	   	p\left(\tilde{\bs}|\tilde{\by}\right) &\propto p_{\textrm{pr}}^{\textrm{c}}\left(\tilde{\bs}\right)f\left(\tilde{\by}|\tilde{\bs}\right)\\
	   	&=\prod_{k=1}^{K}p_{\textrm{pr},k}^{\textrm{c}}\left(\tilde{s}_k\right)f_{\textrm{CG}}\left(\tilde{\by};\tilde{\bG}\tilde{\bs},\sigma_z^2\bI\right),
	\end{split}
\end{equation}
where 
$p_{\textrm{pr}}^{\textrm{c}}\left(\tilde{\bs}\right)$ is the a priori probability distribution of the complex transmitted symbol vector $\tilde{\bs}$,
$f\left(\tilde{\by}|\tilde{\bs}\right)$ is the PDF of the received signal $\tilde{\by}$ given $\tilde{\bs}$,
$p_{\textrm{pr},k}^{\textrm{c}}\left(\tilde{s}_k\right)$ is the a priori probability of the complex symbol transmitted by user $k$,
$p_{\textrm{pr},k}^{\textrm{c}}\left(\tilde{s}_k\right)\big|_{\tilde{s}_k = \tilde{s}^{\left(\ell\right)}} = {1}/{\tilde{L}}, k\in \setposi{K}, \ell \in \setnnega{\tilde{L}-1}$.
Given the \textsl{a posteriori} probability distribution $p\left(\tilde{\bs}|\tilde{\by}\right)$, the  MAP detector (or the ML detector in this case) is  given by
\begin{equation}\label{equ:MAP}
	\tilde{\bs}_{\textrm{MAP}} = \argmax{\tilde{\bs}\in\tilde{\mathbb{S}}^K}p\left(\tilde{\bs}|\tilde{\by}\right),
\end{equation}
which minimizes the error probability that $\tilde{\bs}_{\textrm{MAP}}$ does not coincide with the true one.
The calculation of the MAP detector is unaffordable for practical ultra-massive MIMO systems  since the number of candidates of $\tilde{\bs}$ increases exponentially w.r.t. $K$ and  \eqref{equ:MAP} is NP-hard.

Before proceeding, we reformulate the complex-valued received signal model \eqref{equ:rece model 1}
into a real-valued one, which is necessary for developing IGA-SD in this paper.
Define real vectors 
\begin{subequations}
\begin{equation}\label{equ:real y and z}
	\by \triangleq \left[
	\begin{array}{*{20}{c}}
		\mathcal{R}\braces{\tilde{\by}} \\
		\mathcal{I}\braces{\tilde{\by}} 
	\end{array} 
	\right],
	\bz \triangleq \left[
	\begin{array}{*{20}{c}}
		\mathcal{R}\braces{\tilde{\bz}} \\
		\mathcal{I}\braces{\tilde{\bz}} 
	\end{array} 
	\right] 
	\in \bbR^{2N_r}
\end{equation}
\begin{equation}\label{equ:real s}
		\bs \triangleq \left[
	\begin{array}{*{20}{c}}
		\mathcal{R}\braces{\tilde{\bs}} \\
		\mathcal{I}\braces{\tilde{\bs}} 
	\end{array} 
	\right] 
	\in \bbR^{2K},
\end{equation}
\end{subequations}
and a real matrix
\begin{equation}\label{equ:definition of G}
	\bG \triangleq \left[
	\begin{array}{*{20}{c}}
		\mathcal{R}\braces{\tilde{\bG}}, &-\mathcal{I}\braces{\tilde{\bG}}\\
		\mathcal{I}\braces{\tilde{\bG}} ,&\mathcal{R}\braces{\tilde{\bG}}
	\end{array} 
	\right]\in \bbR^{2N_r\times 2K}.
\end{equation}
Then, we can obtain the real-valued received signal model as
\begin{equation}\label{equ:rece model}
	\by = \bG\bs + \bz,
\end{equation}
where $\bs = \left[ s_1,s_2,\ldots,s_{2K} \right]^T \in \mathbb{S}^{2K}$, 
$s_k \in \mathbb{S}, k\in \setposi{2K}$, 
$\mathbb{S} = \braces{s^{\left(0\right)},s^{\left(1\right)}, \ldots,s^{\left( L-1 \right)}  }$ is the alphabet for the real and imaginary components of a symmetric $\tilde{L}$-QAM modulation where the index starts from $0$ is for the representation convenience of the likelihood ratio detection later, 
$L = \sqrt{\tilde{L}}$,
$\bz \sim \mathcal{N}\left(\mathbf{0},\sigma_z^2\bI\right)$ is the noise vector,
and $\sigma_z^2 = {\tilde{\sigma}_z^2}/{2}$.
Given the received signal model \eqref{equ:rece model}, the \textsl{a posteriori} distribution of $\bs$ can be expressed as
\begin{equation}\label{equ:bayesian 1}
	\begin{split}
		p\left(\bs|\by\right) &\propto \prod_{k=1}^{2K}p_{\textrm{pr},k}\left(s_k\right)\prod_{n=1}^{2N_r}f\left(y_n|\bs\right)\\
		&\propto \prod_{k=1}^{2K}p_{\textrm{pr},k}\left(s_k\right)\prod_{n=1}^{2N_r}\exp\braces{-\frac{\left(y_n-\be_n^T\bG\bs\right)^2}{2\sigma_z^2}},
	\end{split}
\end{equation}
where 
$p_{\textrm{pr},k}\left(s_k\right)\big|_{s_k = s^{\left(\ell\right)}} = \frac{1}{L}, k\in \setposi{2K}, \ell \in \setnnega{L-1}$, is the a priori probability of $s_k$,
$y_n$ is the $n$-th element of $\by$,
$f\left(y_n|\bs\right)$ is the PDF of $y_n, n\in \setposi{2N_r}$, given $\bs$,
and $\be_n \in \bbC^{2N_r}$ is the $n$-th column of the $2N_r$ dimensional identity matrix.
In this work, we propose an information geometry approach for signal detection which aims to obtain the approximations of the marginals, i.e., $p_k\left(s_k|\by\right), k\in \setposi{2K}$,	of the \textsl{a posteriori} distribution $p\left(\bs|\by\right)$, which can be used for the maximization of the \textsl{a posteriori} marginals (MPM) detector, i.e., for $k\in \setposi{2K}$,
\begin{equation}\label{equ:MPM dete}
	s_{k,\textrm{MPM}} = \argmax{s_k\in \mathbb{S}}p_k\left(s_k|\by\right).
\end{equation}

\section{Preliminaries of Information Geometry}\label{sec:pre of IG}
In this section, we briefly introduce the information geometry approach (IGA), where more details can be found in \cite{srbpig,1302292,IGA,IGanditsAPP}.
We begin with the exponential family.
Consider a discrete random vector $\bx \in \mathbb{X}$ with finite dimension, where each scalar random variable in $\bx$ takes finite values and $\mathbb{X}$ is a finite set.
The probability distribution of $\bx$ is said to belong to the exponential family if it can be expressed as 
\begin{equation}\label{equ:exp family}
	p\left(\bx;\btheta\right) = \exp\braces{\btheta^T\bt - \psi\left( \btheta \right)},
\end{equation}
where $\bt$ is a sufficient statistic of random vector $\bx$, 
$\btheta$ is the natural parameter (NP) of $p\left(\bx;\btheta\right)$, 
and $\psi\left(\btheta\right)$ is the free energy, which makes $p\left(\bx;\btheta\right)$ a probability distribution, i.e., $\sum_{\bx}p\left(\bx;\btheta\right) = 1$.
We then introduce the $e$-flat manifold that is needed for m-projections later \cite{srbpig,IGanditsAPP,1302292}.
Consider a manifold $\mathcal{U}$, which is defined as a set of probability distributions of $\bx$, e.g., $\mathcal{U} = \braces{p\left(\bx\right)}$, where each element in $\mathcal{U}$, i.e., $p\left(\bx\right)$, is a particular probability distribution of $\bx$.
$\mathcal{U}$ is said to be $e$-flat if for all $0\le d\le 1$, $p_i\left(\bx\right), p_j\left(\bx\right) \in \mathcal{U}$, the following $q\left(\bx;d\right)$ belongs to $\mathcal{U}$,
\begin{equation}
	q\left(\bx;d\right) = \left(1-d\right)\ln p_i\left(\bx\right) + d\ln p_j\left(\bx\right) + c_{\textrm{na}}\left(d\right),
\end{equation}
where $c_{\textrm{na}}\left(d\right)$ is a normalization constant makes  $q\left(\bx;d\right)$ a probability distribution.
From the definition, any exponential family is $e$-flat.
Suppose  $\mathcal{V} = \braces{q\left(\bx\right)} \subseteq \mathcal{U}$ is an $e$-flat submanifold.
Given $p\left(\bx\right) \in \mathcal{U}$, the point (probability distribution) in $\mathcal{V}$ that minimizes the Kullback-Leibler (K-L) divergence from $p\left(\bx\right)$ to $\mathcal{V}$, i.e.,
\begin{equation}
	q^{\star}\left(\bx\right) = \argmin{q\left(\bx\right)\in \mathcal{V}} \Dkl{p\left(\bx\right)}{q\left(\bx\right)},
\end{equation} 
is called the $m$-projection of $p\left(\bx\right)$ onto $\mathcal{V}$, where the K-L divergence is defined by
\begin{equation}
	\Dkl{p\left(\bx\right)}{q\left(\bx\right)} = \sum_{\bx\in \mathbb{X}}p\left(\bx\right)\ln\left(\frac{p\left(\bx\right)}{q\left(\bx\right)}\right).
\end{equation}

We now give the preliminaries of IGA in Bayesian inference.
Let $\bx_{\textrm{h}} \in \bbR^{N_{\textrm{h}}}$ and $\by_{\textrm{o}} \in \bbR^{N_{\textrm{o}}}$ be hidden and observed random vectors, respectively. 
Denote the \textsl{a posteriori} distribution as $p\left(\bx_{\textrm{h}}|\by_{\textrm{o}}\right)$. 
Our goal is to calculate the approximations
of the \textsl{a posteriori} marginals, i.e., $p\left(x_{\textrm{h},i}|\by_{\textrm{o}}\right)$, where $x_{\textrm{h},i}$ is the $i$-th component of $\bx_{\textrm{h}}$ and $i\in \setposi{N_{\textrm{h}}}$.
In this paper, we focus on the following case:
all the components of $\bx_{\textrm{h}}$ are independent and all the components of $\by_{\textrm{o}}$ given  $\bx_{\textrm{h}}$ are independent  as well.
The  \textsl{a posteriori} distribution can be then expressed as
\begin{equation}
	p\left( \bx_{\textrm{h}}|\by_{\textrm{o}} \right) \propto p\left(\bx_{\textrm{h}}\right)p\left(\by_{\textrm{o}}|\bx_{\textrm{h}}\right) = \prod_{i=1}^{N_{\textrm{h}}}p_i\left(x_{\textrm{h},i}\right)\prod_{n=1}^{N_{\textrm{o}}}p_n\left(y_{\textrm{o},n}|\bx_{\textrm{h}}\right),
\end{equation}
where $p_i\left(x_{\textrm{h},i}\right)$ and $p_n\left(y_{\textrm{o},n}|\bx_{\textrm{h}}\right)$ are the marginals of $p\left(\bx_{\textrm{h}}\right)$ and $p\left(\by_{\textrm{o}}|\bx_{\textrm{h}}\right)$, respectively, and $y_{\textrm{o},n}$ is the $n$-th component  of $\by_{\textrm{o}}$.
Suppose that the a priori marginals $\braces{p_i\left(x_{\textrm{h},i}\right)}_{i=1}^{N_\textrm{h}}$ belong to an exponential family, and each of them can be expressed as
\begin{equation}
	p_i\left(x_{\textrm{h},i}\right) = p_i\left(x_{\textrm{h},i};\bd_{\textrm{h},i}\right) = \exp\braces{\bd_{\textrm{h},i}^T\bt_{\textrm{h},i} - \psi\left( \bd_{\textrm{h},i} \right)},
\end{equation}
where $\bd_{\textrm{h},i}\in \bbR^{N_{i}}$ is the NP of $p_i\left(x_{\textrm{h},i};\bd_{\textrm{h},i}\right)$, 
$\bt_{\textrm{h},i}\in \bbR^{N_{i}}$ is a sufficient statistic of the single random variable $x_{\textrm{h},i}$, e.g., $x_{\textrm{h},i}$ and $x_{\textrm{h},i}^2$,
and $\psi\left( \bd_{\textrm{h},i} \right)$ is the free energy.
As we shall see later in this section, the probability distributions of discrete random vectors belong to the exponential family.
Meanwhile, suppose that the marginals of the conditional probability distribution can be expressed as
\begin{equation}
	p_n\left(y_{\textrm{o},n}|\bx_{\textrm{h}}\right) = \exp\braces{c_n\left( {\bx_{\textrm{h}}},y_{\textrm{o},n} \right) - \psi_n}, n\in \setposi{N_{\textrm{o}}},
\end{equation}
where $c_n\left( {\bx_{\textrm{h}}},y_{\textrm{o},n} \right)$ is a polynomial of $\bx_{\textrm{h}}$ which is parameterized by the variables including $y_{\textrm{o},n}$,
and $\psi_n$ is the normalization factor.
$c_n\left( {\bx_{\textrm{h}}},y_{\textrm{o},n} \right)$ above often contains the interactions between the random variables of $\bx_{\textrm{h}}$, e.g., the cross-terms $x_{\textrm{h},i}x_{\textrm{h},j}, i\neq j$.
A more detailed $c_n\left( {\bx_{\textrm{h}}},y_{\textrm{o},n} \right)$ will occur in the next section.
In this case, the \textsl{a posteriori} probability distribution can be expressed as
\begin{equation}\label{equ:post distribution ori}
	p\left( \bx_{\textrm{h}}|\by_{\textrm{o}} \right) = \exp\braces{\bd_{\textrm{h}}^T\bt_{\textrm{h}} + \sum_{n=1}^{N_{\textrm{o}}}c_n\left(\bx_{\textrm{h}},y_{\textrm{o},n}\right) - \psi_q   },
\end{equation}
where $\bd_{\textrm{h}} = \left[  \bd_{\textrm{h},1}^T, \bd_{\textrm{h},2}^T, \ldots, \bd_{\textrm{h},N_{\textrm{h}}}^T    \right]^T \in \bbR^{N_{a}}$,
$\bt_{\textrm{h}} = \left[  \bt_{\textrm{h},1}^T, \bt_{\textrm{h},2}^T, \ldots, \bt_{\textrm{h},N_{\textrm{h}}}^T    \right]^T \in \bbR^{N_{a}}$,
$N_a = \sum_{i=1}^{N_{\textrm{h}}}N_{i}$,
and $\psi_q$ is the normalization factor.
In \eqref{equ:post distribution ori}, $\bt_{\textrm{h}}$ only contains the separated random variables (i.e., no cross-terms of them), and all the interactions (cross-terms) between the random variables are included in $c_n\left(\bx_{\textrm{h}},y_{\textrm{o},n}\right), n\in \setposi{N_{\textrm{o}}}$.
IGA aims to approximate $\sum_{n=1}^{N_{\textrm{o}}}c_n\left(\bx_{\textrm{h}},y_{\textrm{o},n}\right)$ as $\btheta_{0}^T\bt_{\textrm{h}}$, where $\btheta_{0} \in \bbR^{N_{a}}$, i.e., IGA aims to approximate the summation of all the cross-terms into a summation of non-cross-terms of the random variables, when $N_a$ is large.
In this case, we have
\begin{equation}
		p\left( \bx_{\textrm{h}}|\by_{\textrm{o}} \right) \approx p_0\left(\bx_{\textrm{h}};\btheta_{0}\right) = \exp\braces{\left(\bd_{\textrm{h}} + \btheta_{0}\right)^T\bt_{\textrm{h}} - \psi_0\left(\btheta_{0}\right)},
\end{equation}
where $\psi_0\left(\btheta_{0}\right)$ is the normalization factor.
The marginals of $p_0\left(\bx_{\textrm{h}};\btheta_0\right)$, i.e., $p_0\left( x_{\textrm{h},i};\btheta_{0} \right), i\in \setposi{N_{\textrm{h}}}$, can be calculated easily since $p_0\left(\bx_{\textrm{h}};\btheta_0\right)$ contains no interactions between the random variables $\braces{x_{\textrm{h},i}}_{i=1}^{N_{\textrm{h}}}$. 
To obtain $\btheta_{0}$, we construct two types of manifolds and compute the approximation for each $c_n\left(\bx_{\textrm{h}},y_{\textrm{o},n}\right)$ in an iterative manner, which is 
denoted as $\bxi_{n}^T\bt_{\textrm{h}}$.
At last, $\btheta_{0}$ is calculated as
$\btheta_0 = \sum_{n=1}^{N_\textrm{o}}\bxi_n$.
The two types of manifolds are the objective manifold (OBM) and the auxiliary manifold (AM).
The OBM $\mathcal{M}_0$ is defined as the set of probability distributions of random vector $\bx_{\textrm{h}}$, of which all the components are independent with each other, i.e,
\begin{subequations}\label{equ:OBM}
	\begin{equation}
		\mathcal{M}_0 = \braces{	p_0\left(\bx_{\textrm{h}};\btheta_{0}\right)|\btheta_{0} \in \bbR^{N_a}},
	\end{equation}
	\begin{equation}
		\begin{split}
		p_0\left(\bx_{\textrm{h}};\btheta_{0}\right) &= \prod_{i=1}^{N_{\textrm{h}}}p_{0,i}\left(x_{\textrm{h},i};\btheta_{0,i}\right)\\
		& = \exp\braces{\left(\bd_{\textrm{h}} + \btheta_{0}\right)^T\bt_{\textrm{h}} - \psi_0\left(\btheta_{0}\right)},
		\end{split}
	\end{equation}
	\begin{equation}
		p_{0,i}\left(x_{\textrm{h},i};\btheta_{0,i}\right)\! =\! \exp\braces{\left(\bd_{\textrm{h},i} + \btheta_{0,i}\right)^T\bt_{\textrm{h},i} - \psi_0\left(\btheta_{0,i}\right)},
	\end{equation}
\end{subequations}
where $\btheta_{0} = \left[\btheta_{0,1}^T, \btheta_{0,2}^T, \ldots, \btheta_{0,N_{\textrm{h}}}^T   \right]^T \in \bbR^{N_a}$,
$\btheta_{0,i} \in \bbR^{N_i}$,
$p_{0,i}\left(x_{\textrm{h},i};\btheta_{0,i}\right)$ is the marginal distribution of $p_0\left(\bx_{\textrm{h}};\btheta_{0}\right)$, 
$\psi_0\left( \btheta_{0} \right) = \sum_{i=1}^{N_\textrm{h}} \psi_0\left(\btheta_{0,i}\right)$ is the free energy (normalization factor) of $p_0\left(\bx_{\textrm{h}};\btheta_{0}\right)$,
and $ \psi_0\left(\btheta_{0,i}\right) $ is the free energy of $ p_{0,i}\left(x_{\textrm{h},i};\btheta_{0,i}\right) $.
$\btheta_{0}$ above is referred as to the $e$-affine coordinate system or the natural parameter of $p_0\left(\bx_{\textrm{h}};\btheta_{0}\right)$. 
And 
$\btheta_{0,i}$ is referred as to the $e$-affine coordinate system or the natural parameter of $p_{0,i}\left(x_{\textrm{h},i};\btheta_{0,i}\right)$.
To avoid confusion with the natural parameter of the exponential family, we refer to $\btheta_{0}$ as the $e$-affine coordinate system (abbreviated as EACS) of $p_0\left(\bx_{\textrm{h}};\btheta_{0}\right)$ in this paper (similar with $\btheta_{0,i}$ and $p_{0,i}\left(x_{\textrm{h},i};\btheta_{0,i}\right)$).
Then, $N_{\textrm{o}}$ AMs are defined, where the $n$-th of them is expressed as
\begin{subequations}
	\begin{equation}
		\mathcal{M}_n = \braces{p_n\left(\bx_{\textrm{h}};\btheta_{n}\right)|\btheta_{n}\in \bbR^{N_a}},
	\end{equation}
	\begin{equation}\label{equ:pn preli}
		p_n\left(\bx_{\textrm{h}};\btheta_{n}\right) \!=\! \exp\braces{\left(\bd_{\textrm{h}} \!+\! \btheta_{n}\right)^T\bt_{\textrm{h}} \!+\! c_n\left(\bx_{\textrm{h}},y_{\textrm{o},n}\right) \!-\! \psi_n\left(\btheta_{n}\right)},
	\end{equation}
\end{subequations}
where $\btheta_{n}$ is referred as to the EACS of $p_n\left( \bx_{\textrm{h}};\btheta_{n} \right)$ and $\psi_n\left(\btheta_{n}\right)$ is the free energy.
It can be readily checked that the OBM and the AMs are all $e$-flat.
Only one interaction term $c_n\left(\bx_{\textrm{h}},y_{\textrm{o},n}\right)$ is remained in $p_n\left(\bx_{\textrm{h}};\btheta_{n}\right)$, and all the others, i.e., $\sum_{n'\neq n}c_{n'}\left(\bx_{\textrm{h}},y_{\textrm{o},n'}\right)$ are replaced as $\btheta_n^T\bt_{\textrm{h}}$.
Assume that the EACS $\btheta_{n}$ of $p_n\left(\bx_{\textrm{h}};\btheta_{n}\right), n\in \setposi{N_\textrm{o}}$, is given, we calculate the approximation of $c_n\left(\bx_{\textrm{h}},y_{\textrm{o},n}  \right)$ from the $m$-projection of $p_n\left(\bx_{\textrm{h}};\btheta_{n}\right)$ onto the OBM $\mathcal{M}_0$.
Denote the $m$-projection of $p_n\left(\bx_{\textrm{h}};\btheta_{n}\right)$ onto $\mathcal{M}_0$ as $p_0\left(\bx_{\textrm{h}};\btheta_{0n}\right)$, where $\btheta_{0n} \in \bbR^{N_a}$, and
\begin{equation}\label{equ:theta_{0n} in preli}
	\btheta_{0n} = \argmin{\btheta_{0}\in \bbR^{N_a}}\Dkl{p_n\left(\bx_{\textrm{h}};\btheta_{n}\right)}{p_0\left(\bx_{\textrm{h}};\btheta_{0}\right)}.
\end{equation}
We shall see a more specific example about the calculation of the $m$-projection in the next section.
After $\btheta_{0n}$ is obtained,
we express the $m$-projection $p_0\left(\bx_{\textrm{h}};\btheta_{0n}\right)$ as
\begin{equation}\label{equ:mp preli}
	\begin{split}
	   	p_0\left(\bx_{\textrm{h}};\btheta_{0n}\right) &= \exp\braces{\left( \bd_{\textrm{h}}+\btheta_{0n} \right)^T\bt_{\textrm{h}} - \psi_0\left(\btheta_{0n}\right)}\\
	   	&= \exp\braces{\left( \bd_{\textrm{h}}+\btheta_{n}+\bxi_{n} \right)^T\bt_{\textrm{h}} - \psi_0\left(\btheta_{0n}\right)},
	\end{split}
\end{equation}
where the EACS $\btheta_{0n}$ of $p_0\left(\bx_{\textrm{h}};\btheta_{0n}\right)$ is regarded as the sum of the EACS $\btheta_{n}$ of $p_n\left(\bx_{\textrm{h}};\btheta_{n}\right)$ and an extra item $\bxi_{n}$.
If we compare the last equation of \eqref{equ:mp preli} and $p_n\left(\bx_{\textrm{h}};\btheta_{n}\right)$ in \eqref{equ:pn preli}, it can be found that in the $m$-projection $p_0\left(\bx_{\textrm{h}};\btheta_{0n}\right)$, the interaction item $c_n\left(\bx_{\textrm{h}},y_{\textrm{o},n}\right)$ is replaced by $\bxi_{n}^t\bt_{\textrm{h}}$.
Hence, $\bxi_{n}^T\bt_{\textrm{h}}$ is regarded as the approximation of $c_n\left(\bx_{\textrm{h}},y_{\textrm{o},n}\right)$, and we calculate the approximation item $\bxi_{n}$ as
\begin{equation}\label{equ:xi_n in preli}
	\bxi_{n} = \btheta_{0n} - \btheta_{n}, n\in \setposi{N_{\textrm{o}}}.
\end{equation}
Then, $p_0\left(\bx_{\textrm{h}};\btheta_{0}\right)$ with $\btheta_{0} = \sum_{n=1}^{N_{\textrm{o}}}\bxi_{n}$ is considered as the approximation of the \textsl{a posteriori} distribution $p\left(\bx_{\textrm{h}}|\by_{\textrm{o}}\right)$.
Meanwhile, note that the whole process is proceeded in an iterative manner since the EACSs $\braces{\btheta_{n}}_{n=1}^{N_\textrm{o}}$ are not known at first.
To be specific, we first  initialize  the EACSs as $ \braces{\btheta_{n}\left(0\right)}_{n=0}^{N_\textrm{o}}$.
Given the EACS $\btheta_{0}\left(t\right)$ of $p_0\left(\bx_{\textrm{h}};\btheta_{0}\left(t\right)\right)$ and the EACS $\btheta_{n}\left(t\right)$ of $p_n\left(\bx_{\textrm{h}};\btheta_{n}\left(t\right)\right), n\in \setposi{N_{\textrm{o}}}$, at the $t$-th time, we  calculate $\btheta_{0n}\left(t\right)$ and $\bxi_{n}\left(t\right), n\in \setposi{N_\textrm{o}}$, as \eqref{equ:theta_{0n} in preli} and \eqref{equ:xi_n in preli}, respectively.
We then update the EACS of $p_n\left(\bx_{\textrm{h}};\btheta_{n}\left(t\right)\right), n\in \setposi{N_\textrm{o}}$, as 
\begin{equation}
	\btheta_{n}\left(t+1\right) = \sum_{n' = 1, n' \neq n}^{N_\textrm{o}}\bxi_{n'}\left(t\right),
\end{equation}
since  $\btheta_{n}^T\left(t+1\right)\bt_{\textrm{h}}$ replaces $\sum_{n' \neq n}c_{n'}\left(\bx_{\textrm{h}},y_{\textrm{o},n}\right)$ in $p_n\left(\bx_{\textrm{h}};\btheta_{n}\left(t+1\right)\right)$ and each interaction term $c_n\left( \bx_{\textrm{h}},y_{\textrm{o},n} \right)$ is approximated as $\bxi_{n}^T\left(t\right)\bt_{\textrm{h}}$ at the $t$-th time.
The EACS of $p_0\left(\bx_{\textrm{h}};\btheta_{0}\left(t\right)\right)$ is updated as $\btheta_{0}\left(t+1\right) = \sum_{n=1}^{N_\textrm{o}}\bxi_{n}\left(t\right)$ as mentioned above.
Then, repeat the $m$-projection, calculate the approximation terms $\braces{\bxi_{n}}_{n=1}^{N_\textrm{o}}$ and the updates until convergence.
We now discuss about the damped updating.
In practice, to improve the convergence of the IGA, the EACSs $\braces{\btheta_{n}}_{n=0}^{N_\textrm{o}}$ are usually updated in a damped way, i.e.,
\begin{subequations}\label{equ:update of NPs}
	\begin{equation}
		\btheta_{n}\left(t+1\right) = \alpha \sum_{n' = 1, n'\neq n}^{N_\textrm{o}}\bxi_{n'}\left(t\right) + \left(1-\alpha\right)\btheta_{n}\left(t\right), n\in \setposi{N_\textrm{o}},
	\end{equation}
\begin{equation}
	\btheta_{0}\left(t+1\right) = \alpha\sum_{n=1}^{N_\textrm{o}}\bxi_n\left(t\right)+\left(1-\alpha\right)\btheta_{0}\left(t\right),
\end{equation}
\end{subequations}
where $0<\alpha\le 1$ is the damping.

At the end of this section, we formulate a  probability distribution of discrete random vectors as one in  the exponential family.
Consider an $N$ dimensional discrete random vector $\bx \in \mathbb{X}$, where 
each component of $\bx$ takes only finite values,
$\mathbb{X}=\braces{\bx^{\left(0\right)}, \bx^{\left(1\right)},\ldots,\bx^{\left(N_{\textrm{x}}-1\right)}}$,
and $N_{\textrm{x}} \ge 2$ is the number of all possible vectors of $\bx$. 
Denote the probability distribution of $\bx$ as $p\left(\bx\right)$ and the probability of $\bx$ taking the value $\bx^{\left(i\right)}$ as $p\left(\bx\right)\big|_{\bx = \bx^{\left(i\right)}} = p_i >0, i\in \setnnega{N_\textrm{x}-1}$.
Denote the set of probability distributions of $\bx$ as
\begin{equation}
	\mathcal{X} = \braces{p\left(\bx\right)\Big|p\left(\bx\right)>0, \bx\in \mathbb{X},  \sum_{\bx\in\mathbb{X}}p\left(\bx\right) = 1}.
\end{equation}
For the discrete probability distributions, let
\begin{equation}
	t_{\textrm{x},i} = \delta\left(\bx-\bx^{\left(i\right)}\right)=
	\begin{cases}
		1, &\textrm{when}\quad  \bx = \bx^{\left(i\right)},\\
		0, &\textrm{otherwise},
	\end{cases}
\end{equation}
where $i\in \setnnega{N_{\textrm{x}-1}}$. 
Then, the probability distribution of $\bx$ can be rewritten as 
\begin{align}\label{equ:express of px}
		p\left(\bx\right) = \sum_{\bx\in \mathbb{X}} p\left(\bx\right)\big|_{\bx = \bx^{\left(i\right)}}\delta\left(\bx-\bx^{\left(i\right)}\right)
		=\sum_{i=0}^{N_\textrm{x}-1}p_i t_{\textrm{x},i},
\end{align}
where $\braces{p_i}_{i=0}^{N_\textrm{x}-1}$ are positive values and  constrained by $\sum_{i=0}^{N_\textrm{x}-1}p_i=1$. Hence, $\mathcal{X}$ has $N_\textrm{x}-1$ degrees of freedom and is a $N_\textrm{x}-1$ dimensional manifold \cite{1302292}.
Since the dimension of $\mathcal{X}$ is $N_\textrm{x}-1$, we define an $N_\textrm{x}-1$ dimensional parameter vector as $\btheta_\textrm{x} = \left[ \theta_{\textrm{x},1},\theta_{\textrm{x},2},\ldots,\theta_{\textrm{x},N_\textrm{x}-1} \right]^T$, where each component is given by
\begin{equation}
	\theta_{\textrm{x},i} = \ln\left(  \frac{p_i}{p_0} \right), i\in \setposi{N_\textrm{x}-1}.
\end{equation}
Then,
\begin{equation}
		p\left(\bx\right) 
		=\exp\braces{\btheta_{\textrm{x}}^T\bt_{\textrm{x}}- \psi\left(\btheta_{\textrm{x}}\right) },
\end{equation}
where $\bt_{\textrm{x}} = \left[t_{\textrm{x},1}, t_{\textrm{x},2},\ldots,t_{\textrm{x},N_\textrm{x}-1}  \right]^T$ is a random vector of $N_\textrm{x}-1$ dimension, and
\begin{equation}
	\psi\left(\btheta_{\textrm{x}}\right) = -\ln p_0.
\end{equation}
The above expresses $\mathcal{X}$ is expressed in terms of an exponential family, and $\btheta_{\textrm{x}}$ is the NP of $p\left(\bx\right)$.

\section{Information Geometry Approach for Signal Detection}

\begin{algorithm}[h]
	\SetAlgoNoLine 
	\caption{IGA-SD}
	\label{Alg:IGA}
	
	\KwIn{The a priori probability $p_{\textrm{pr},k}\left(s_k\right), k\in \setposi{2K}$, the received signal $\by$, the channel matrix $\bG$, the alphabet $\mathbb{S} = \braces{s^{\left(0\right)},s^{\left(1\right)}, \ldots,s^{\left( L-1 \right)}}$ for the components of $\bs$, the noise power ${\sigma}_z^2$ and the maximal iteration number $t_{\mathrm{max}}$.}
	
	\textbf{Initialization:} set $t=0$, set damping $\alpha$, where $0< \alpha\le 1$, initialize the EACSs $\btheta_{n}, n\in \setnnega{2N_r}$, which are defined in \eqref{equ:EACS of p0} and \eqref{equ:EACS of pn}, zeros are sufficient for their initializations in general, calculate the NP $d_{k,\ell}, k\in \setposi{2K}, \ell \in \setposi{L-1}$, as \eqref{equ:d_{k,l}}; 
	
	\Repeat{\rm{Convergence or $t > t_{\mathrm{max}}$}}{
		1. Calculate $\bxi_n{\left(t\right)}, n\in \setposi{2N_r}$, as \eqref{equ:xi_{n,k}} and \eqref{equ:xi_{n,k,l}};\\
		2. Update the EACSs as
		\eqref{equ:update of NPs IGA};\\
		3. $t = t+1$;}
	
	\KwOut{\rm{The probability of the approximate marginal, $p_k\left( s_k|\by\right)$, is given by the probability of $p_{0,k}\left(s_k;\btheta_{0,k}\right)$, $k\in \setposi{2K}$, which is given by \eqref{equ:probability of marginals of p0}. Then, the MPM detection is given by \eqref{equ:MPM dete}.}}
\end{algorithm}

As discussed in Sec. \ref{sec:pre of IG},  $p_{\textrm{pr},k}\left(s_k\right), k\in \setposi{2K}$, belong to the exponential family.
Define a sufficient statistic as $\bt_k \triangleq \left[t_{k,1},t_{k,2},\ldots,t_{k,L-1}  \right]^T \in \bbR^{\left(L-1\right)}$, 
where $k\in \setposi{2K}, \ell \in \setposi{L-1}$.
Define the NP as $\bd_k \triangleq \left[ d_{k,1},d_{k,2},\ldots,d_{k,L-1} \right]^T \in \bbR^{\left(L-1\right)}, k\in \setposi{2K}$, and
\begin{equation}\label{equ:d_{k,l}}
	d_{k,\ell} =\ln \frac{p_{\textrm{pr},k}\left(s_k\right)\big|_{s_k = s^{\left(\ell\right)}}}{p_{\textrm{pr},k}\left(s_k\right)\big|_{s_k = s^{\left(0\right)}}}, \ell \in \setposi{L-1}.
\end{equation}
Then, $p_{\textrm{pr},k}\left(s_k\right), k\in \setposi{2K}$, can be expressed as
\begin{equation}\label{equ:p_k in exp family}
	p_{\textrm{pr},k}\left(s_k\right) = \exp\braces{\bd_k^T\bt_k - \psi\left( \bd_k \right)},
\end{equation}
where $\psi\left( \bd_k \right) = -\ln \left(p_{\textrm{pr},k}\left(s_k\right)\big|_{s_k = s^{\left(0\right)}}  \right)$ is the free energy.
Combining with \eqref{equ:p_k in exp family}, the \textsl{a posteriori} distribution $p\left(\bs|\by\right)$ can be expressed as
\begin{align}\label{equ:post distribution}
	p\left(\bs|\by\right) &= \exp\braces{\sum_{k=1}^{2K}\bd_k^T\bt_k + \sum_{n=1}^{2N_r}c_n\left( \bs,y_n \right)  - \psi_q  }\nonumber\\
	&= \exp\braces{  \bd^T\bt + \sum_{n=1}^{2N_r}c_n\left( \bs,y_n \right) - \psi_q },
\end{align}
where $\bd = \left[ \bd_1^T,\bd_2^T,\ldots,\bd_{2K}^T,  \right]^T \in \bbR^{2K\left(L-1\right)}$, $\bt = \left[ \bt_1^T,\bt_2^T,\ldots,\bt_{2K}^T \right]^T \in \bbR^{2K\left(L-1\right)}$, $\psi_q$ is the normalization factor, and
\begin{subequations}
	\begin{equation}
		c_n\left( \bs,y_n \right) = -\frac{1}{2\sigma_z^2}\left( y_n - \be_n^T\bG\bs \right)^2,
	\end{equation}
\begin{equation}
	\psi_q = \ln \left( \sum_{\bs \in \mathbb{S}^{2K}}\exp\braces{ \bd^T\bt +  \sum_{n=1}^{2N_r}c_n\left( \bs,y_n \right) } \right).
\end{equation}
\end{subequations}
According to \eqref{equ:post distribution}, we can immediately define the OBM and the AMs as in the previous section.
The OBM is defined as
\begin{subequations}
	\begin{equation}
		\mathcal{M}_0 = \braces{p_0\left(\bs;\btheta_0\right)\Big|\btheta_0\in \bbR^{2K\left(L-1\right) }  },
	\end{equation}
\begin{equation}\label{equ:defintion of p_0}
	\begin{split}
	p_0\left(\bs;\btheta_0\right) &= \prod_{k=1}^{2K}p_{0,k}\left(s_k;\btheta_{0,k}\right)\\
	&= \exp\braces{\bd^T\bt + \btheta_0^T\bt - \psi_0\left(\btheta_0\right)},
	\end{split}
\end{equation}
\begin{align}\label{equ:marginals of p0}
	   	p_{0,k}\left(s_k;\btheta_{0,k}\right) &=  \exp\braces{\bd_k^T\bt_k + \btheta_{0,k}^T\bt_k - \psi_0\left(\btheta_{0,k}\right)}\nonumber\\
        &=\exp\braces{\sum_{\ell=1}^{L-1}\left( d_{k,\ell} + \theta_{0,k,\ell}\right)\delta\left(s_{k} - s^{\left(\ell\right)}\right)  } \nonumber\\
        & \ \ \times \exp\braces{- \psi_0\left(\btheta_{0,k}\right)  },
\end{align}
\end{subequations}
where
\begin{equation}\label{equ:EACS of p0}
	\btheta_0 = \left[\btheta_{0,1}^T,\btheta_{0,2}^T,\ldots,\btheta_{0,2K}^T  \right]^T \in \bbR^{2K\left(L-1\right)}
\end{equation}
is the EACS of $p_0\left(\bs;\btheta_0\right)$,
\begin{equation}
	\btheta_{0,k} = \left[ \theta_{0,k,1},\theta_{0,k,2},\ldots,\theta_{0,k,L-1}  \right]^T\in \bbR^{\left(L-1\right)}
\end{equation}
is the EACS of $p_{0,k}\left(s_k;\btheta_{0,k}\right)$,
$p_{0,k}\left(s_k;\btheta_{0,k}\right)$ is the marginal
distribution of $s_k$, 
the free energies $\psi_0\left(\btheta_{0}\right)$ and $\psi_0\left(\btheta_{0,k}\right)$ are given by 
\begin{subequations}\label{equ:psi0}
	\begin{equation}
		\begin{split}
			\psi_0\left(\btheta_0\right) &= \sum_{k=1}^{K}\psi_0\left(\btheta_{0,k}\right)\\
			&=\ln \left( \sum_{\bs \in \mathbb{S}^{2K}}\exp\braces{ \bd^T\bt + \btheta_{0}^T\bt } \right),
		\end{split}
	\end{equation}
\begin{equation}\label{equ:free energy of marginals of p0}
	\begin{split}
	   \psi_0\left(\btheta_{0,k}\right) &= \ln \left( \sum_{s_k \in \mathbb{S}}\exp\braces{ \bd_k^T\bt_k + \btheta_{0,k}^T\bt_k } \right)\\
	   &=\ln\left( 1+ \sum_{\ell = 1}^{L-1}\exp\braces{ d_{k,\ell} + \theta_{0,k,\ell} } \right).
	\end{split}
\end{equation}
\end{subequations}
Given $p_0\left(\bs;\btheta_{0}\right)$ and its marginals $p_{0,k}\left(s_k;\btheta_{0,k}\right)$, the probability of signal $s_k, k\in \setposi{2K}$, can be expressed in a more explicit way as
\begin{subequations}\label{equ:probability of marginals of p0}
	\begin{equation}\label{equ:probability of mariginals of p0 on s^0}
	p_{0,k}\left(s_k;\btheta_{0,k}\right)\Big|_{s_{k} = s^{\left(0\right)}} \equaa \frac{1}{1+ \sum_{\ell = 1}^{L-1}\exp\braces{ d_{k,\ell} + \theta_{0,k,\ell} }},
    \end{equation}
    \begin{equation}\label{equ:probability of mariginals of p0 on s^ell}
	p_{0,k}\left(s_k;\btheta_{0,k}\right)\Big|_{s_{k} = s^{\left(\ell\right)}} \equab \frac{\exp\braces{ d_{k,\ell} + \theta_{0,k,\ell} }}{1+ \sum_{\ell = 1}^{L-1}\exp\braces{ d_{k,\ell} + \theta_{0,k,\ell} }},
    \end{equation}
\end{subequations}
where $\ell \in \setposi{L-1}$ in \eqref{equ:probability of mariginals of p0 on s^ell}, and $\left(\textrm{a}\right)$ and $\left(\textrm{b}\right)$ come from \eqref{equ:marginals of p0} and \eqref{equ:free energy of marginals of p0}.
The probability of $p_0\left(\bs;\btheta_{0}\right)$ can be then expressed more explicitly by using \eqref{equ:defintion of p_0}.
Also, from \eqref{equ:probability of marginals of p0}, we can conversely use the marginal probability of $s_k$ to express the EACS $\btheta_{0,k}$ of $p_{0,k}\left(s_k;\btheta_{0,k}\right), k\in \setposi{2K}$, i.e.,
\begin{equation}\label{equ:relation}
	\theta_{0,k,\ell} = \ln\frac{p_{0,k}\left(s_k;\btheta_{0,k}\right)\Big|_{s_{k} = s^{\left(\ell\right)}}}{p_{0,k}\left(s_k;\btheta_{0,k}\right)\Big|_{s_{k} = s^{\left(0\right)}}} - d_{k,\ell}, \ell\in \setposi{L-1}.
\end{equation}
Then, the EACS $\btheta_{0}$ of $p_0\left(\bs;\btheta_{0}\right)$ can be also obtained.
This relationship will be used later in this section.
$2N_r$ AMs are defined, where the $n$-th of them is given by
\begin{subequations}\label{equ:AMs}
\begin{equation}
	\mathcal{M}_n = \braces{ p_n\left(\bs;\btheta_n\right) \Big| \btheta_n \in \bbR^{2K\left(L-1\right)}  },
\end{equation}
\begin{equation}\label{equ:pn}
	p_n\left(\bs;\btheta_n\right) = \exp\braces{ \bd^T\bt + \btheta_n^T\bt + c_n\left(\bs,y_n\right) - \psi_n\left(\btheta_n\right)  },
\end{equation}
\end{subequations}
where
\begin{equation}\label{equ:EACS of pn}
	\btheta_n = \left[\btheta_{n,1}^T,\btheta_{n,2}^T,\ldots,\btheta_{n,2K}^T  \right]^T \in \bbR^{2K\left(L-1\right)}
\end{equation}
is the EACS of $p_n\left(\bs;\btheta_n\right)$,
\begin{equation}
	\btheta_{n,k} = \left[ \theta_{n,k,1},\theta_{n,k,2},\ldots,\theta_{n,k,L-1}  \right]^T\in \bbR^{\left(L-1\right)},
\end{equation}
and the free energy $\psi_n$ is given by 
\begin{equation}
	\psi_n\left(\btheta_n\right) = \ln \left( \sum_{\bs \in \mathbb{S}^{2K}}\exp\braces{ \bd^T\bt + \btheta_n^T\bt + c_n\left( \bs,y_n \right) } \right).
\end{equation}
From the definitions, it is not difficult to check that the OBM and the AMs are all $e$-flat.

Before proceeding, we further define a manifold called the original manifold (OM), and then show that the OBM and the AMs are its submanifolds.
Define the OM as the set of probability distributions of the $2K$ dimensional discrete random vector $\bs$ as
\begin{equation}
	\mathcal{S} = \braces{p\left(\bs\right)\Big| p\left(\bs\right)>0, \bs\in \mathbb{S}^{2K},  \sum_{\bs\in\mathbb{S}^{2K}}p\left(\bs\right) = 1 }.
\end{equation}
$\mathcal{S}$ is then a $L^{2K}-1$ dimensional  manifold and forms an exponential family.
Then, it can be readily checked that the \textsl{a posteriori} distribution $p\left(\bs|\by
\right)$ belongs to $\mathcal{S}$ since $p\left(\bs|\by\right)$  is a particular probability distribution of $\bs$.
Similarly, it can be obtained that the OBM and the AMs are the submanifolds of the OM, i.e., $\mathcal{M}_0 \subseteq \mathcal{S}$, $\mathcal{M}_n \subseteq \mathcal{S}, n\in\setposi{2N_r}$, since the distributions in the OBM and the AMs are all particular probability distributions of $\bs$ when the EACSs of them are given.

We now present the properties of the $m$-projection of any $p\left(\bs\right) \in \mathcal{S}$, onto the OBM $\mathcal{M}_0$, which inspires us  to approximate the $m$-projection of $p_n\left(\bs;\btheta_{n}\right)$ onto the OBM $\mathcal{M}_0$. 
According to the Section \ref{sec:pre of IG}, given $p\left(\bs\right) \in \mathcal{S}$ and the OBM $\mathcal{M}_0$, which is an $e$-flat submanifold of $\mathcal{S}$,  the $m$-projection of $p\left(\bs\right)$ onto $\mathcal{M}_0$  is obtained by the following minimization problem,
	\begin{equation}\label{equ:mp in detection 1}
		\btheta_{0}^{\star} = \argmin{\btheta_0}\Dkl{p\left(\bs\right)}{p_0\left( \bs;\btheta_0 \right)},
	\end{equation}
where the K-L divergence is given by
\begin{align}\label{equ:K-L divergence in detection}
	&\ \Dkl{p\left(\bs\right)}{p_0\left( \bs;\btheta_0 \right)}
	= \Exp_{p\left(\bs\right)} \braces{  \ln \frac{p\left(\bs\right)}{p_0\left(\bs;\btheta_0\right)}   }\nonumber \\
	=&\ C_{p} - \sum_{\bs \in \mathbb{S}^{2K}}p\left(\bs\right)\ln \left( p_0\left(\bs;\btheta_0\right) \right),
\end{align}
where $C_{p} = \sum_{\bs \in \mathbb{S}^{2K}}p\left(\bs\right)\ln p\left(\bs\right)$ is a constant independent of $\btheta_0$.
We then have the following theorem.
\begin{theorem}\label{lemma:equivalent of mp1}
	Given $p\left( \bs \right) \in \mathcal{S} $, and the $e$-flat $\mathcal{M}_0 \subseteq \mathcal{S}$,
	the $m$-projection of $p\left( \bs \right)$ onto $\mathcal{M}_0$ is unique.
	Moreover, $p_0\left(\bs;\btheta_{0}^{\star}\right)$ is the $m$-projection of $p\left( \bs \right)$ onto $\mathcal{M}_0$ if and only if  the following relationship holds,
	\begin{equation}\label{equ:mp invariant1}
	   \bbeta = \bbeta_0\left(\btheta_{0}^{\star}\right) ,
	\end{equation}
where $\bbeta, \bbeta_0\left(\btheta_{0}^{\star}\right) \in \bbR^{2K\left( L-1 \right)} $ are the expectations of $\bt$ w.r.t. $p\left(\bs\right)$ and $p_0\left(\bs;\btheta_{0}^{\star}\right)$, respectively, i.e.,
\begin{subequations}
	\begin{equation}
		\bbeta = \Exp_{p\left(\bs\right)}\braces{\bt} = \sum_{\bs \in \mathbb{S}^{2K}}\bt p\left(\bs\right),
	\end{equation}
\begin{equation}
	\bbeta_0\left(\btheta_{0}^{\star}\right) = \Exp_{p_0\left(\bs;\btheta_{0}^{\star}\right)}\braces{\bt} = \sum_{\bs \in \mathbb{S}^{2K}}\bt p_0\left(\bs;\btheta_{0}^{\star}\right).
\end{equation}
\end{subequations}
\end{theorem}
\begin{proof}
	See Appendix \ref{proof:equivalent of mp1}.
\end{proof}

Define $2K$ discrete random vectors of $2K-1$ dimensions, where the $k$-th of them, denoted as $\bs_{\setminus k}$, is obtained by  removing the $k$-th element, i.e., $s_k$, of $\bs$, $ k\in \setposi{2K}$.
Then, we can obtain $\bs_{\setminus k}\in \mathbb{S}^{2K-1}, k\in \setposi{2K}$, and the marginal probability distribution of $s_k$ given the joint probability distribution $p\left(\bs\right)$  is
	\begin{equation}\label{equ:marginals of p(s)}
		\begin{split}
		p_k\left(s_k\right) &\triangleq \sum_{\bs_{\setminus k} \in \mathbb{S}^{2K-1}}p\left(\bs\right)\\
		&=\sum_{s_1\in \mathbb{S}}\cdots \sum_{s_{k-1}\in \mathbb{S}}\sum_{s_{k+1}\in \mathbb{S}}\cdots\sum_{s_{2K}\in \mathbb{S}} p\left(\bs\right),  k\in \setposi{2K}.
		\end{split}
\end{equation}
From the definition of $p_0\left(\bs;\btheta_{0}\right)$ in \eqref{equ:defintion of p_0}, we denote the marginals of $p_0\left(\bs;\btheta_{0}^{\star}\right)$ in Theorem \ref{lemma:equivalent of mp1} as $p_{0,k}\left(s_k;\btheta_{0,k}^{\star}\right), k\in \setposi{2K}$, where
$	\btheta_{0,k}^{\star} = \left[ \theta_{0,k,1}^{\star},\theta_{0,k,2}^{\star},\ldots,\theta_{0,k,L-1}^{\star}  \right]^T\in \bbR^{\left(L-1\right)}$ and
$\btheta_{0}^{\star} = \left[ \left( \btheta_{0,1}^{\star} \right)^T, \left( \btheta_{0,2}^{\star} \right)^T, \ldots, \left( \btheta_{0,2K}^{\star} \right)^T \right]^T$.
Combining Theorem \ref{lemma:equivalent of mp1}, we have the following corollary.
\begin{corol}\label{the:equivalent of mp2}
	Given $p\left( \bs \right) \in \mathcal{S} $, and the $e$-flat $\mathcal{M}_0 \subseteq \mathcal{S}$, $p_0\left(\bs;\btheta_{0}^{\star}\right)$ is the $m$-projection of $p\left( \bs \right)$ onto $\mathcal{M}_0$ if and only if the marginals of $p\left(\bs\right)$ and the marginals of $p_0\left(\bs;\btheta_{0}^{\star}\right)$ are equal, i.e.,
	\begin{equation}\label{equ:mp3}
		p_k\left(s_k\right) = p_{0,k}\left(s_k;\btheta_{0,k}^{\star} \right), s_k\in \mathbb{S}, k\in \setposi{2K}.
	\end{equation}
Meanwhile, the EACS of the $m$-projection is given by $\btheta_{0}^{\star} = \left[ \left( \btheta_{0,1}^{\star} \right)^T, \left( \btheta_{0,2}^{\star} \right)^T, \ldots, \left( \btheta_{0,2K}^{\star} \right)^T \right]^T$, where  $\btheta_{0,k}^{\star} = \left[ \theta_{0,k,1}^{\star},\theta_{0,k,2}^{\star},\ldots,\theta_{0,k,L-1}^{\star}  \right]^T, k\in \setposi{2K}$, and
	\begin{equation}\label{equ:mp solution 0}
		\theta_{0,k,\ell}^{\star} = \ln\frac{p_k\left(s_k\right)\big|_{s_k = s^{\left(\ell\right)}}}{p_k\left(s_k\right)\big|_{s_k = s^{\left(0\right)}}} - d_{k,\ell}, \ell\in \setposi{L-1}.
	\end{equation}
\end{corol}
\begin{proof}
	See  Appendix \ref{proof:equivalent of mp2}.
\end{proof}

Given $p_n\left(\bs;\btheta_{n}\right) \in \mathcal{M}_n, n\in \setposi{2N_r}$, from Theorem \ref{lemma:equivalent of mp1} we can obtain that its $m$-projection  onto $\mathcal{M}_0$ is unique since $p_n\left(\bs;\btheta_{n}\right)\in \mathcal{S}$.
Denote the $m$-projection of $p_n\left(\bs;\btheta_{n}\right)$ onto $\mathcal{M}_0$ as $p_0\left(\bs;\btheta_{0n}\right), n\in \setposi{2N_r}$, where $\btheta_{0n} = \left[ \btheta_{0n,1}^T,\btheta^T_{0n,2},\ldots,\btheta^T_{0n,2K} \right]^T \in \bbR^{2K\left(L-1\right)}$ and $\btheta_{0n,k}= \left[ \theta_{0n,k,1}, \theta_{0n,k,2}, \ldots, \theta_{0n,k,L-1} \right]^T \in \bbR^{\left(L-1\right)}, k\in \setposi{2K}$.
From Corollary \ref{the:equivalent of mp2}, for any $n\in \setposi{2N_r}$, the $m$-projection  $p_0\left(\bs;\btheta_{0n}\right)$ is determined by the marginal probability distribution $p_{n,k}\left(s_k;\btheta_{n}\right), k \in \setposi{2K}$, where
\begin{equation}
	p_{n,k}\left(s_k;\btheta_{n}\right) \triangleq \sum_{\bs_{\setminus k} \in \mathbb{S}^{2K-1}}p_n\left(\bs;\btheta_n\right).
\end{equation}
And we have 
\begin{equation}\label{equ:mp solution}
	\theta_{0n,k,\ell} = \ln\frac{p_{n,k}\left(s_k;\btheta_{n}\right)\big|_{s_k = s^{\left(\ell\right)}}}{p_{n,k}\left(s_k;\btheta_{n}\right)\big|_{s_k = s^{\left(0\right)}}} - d_{k,\ell},
\end{equation}
where $n\in \setposi{2N_r}, k\in \setposi{2K}$ and $\ell \in \setposi{L-1}$.
Nevertheless, it is relatively difficult to obtain the closed-form solution of the marginal probability distribution $p_{n,k}\left(s_k;\btheta_{n}\right)$ since the calculation is of exponential-complexity.
In this work, we solve this problem by calculating the approximations of the marginals $p_{n,k}\left(s_k;\btheta_{n}\right), k\in \setposi{2K}, n\in \setposi{2N_r}$, using the central-limit-theorem (CLT).

From the definition of $p_n\left(\bs;\btheta_{n}\right)$ in \eqref{equ:pn}, its marginals can be expressed as
\begin{align}\label{equ:marginals of pn}
		p_{n,k}\left(s_k;\btheta_n\right) &= \sum_{\bs_{\setminus k} \in \mathbb{S}^{2K-1}}\exp\braces{ \left(\bd +\btheta_n  \right)^T \bt + c_n\left(\bs,y_n\right) - \psi_n }\nonumber\\
		&\overset{\left(\textrm{a}\right) }{\propto}\exp\braces{\left( \bd_k+\btheta_{n,k} \right)^T\bt_k }q\left( y_n,s_k \right),
\end{align}
where $n \in \setposi{2N_r}$, 
$k\in \setposi{2K}$, 
$s_k \in \mathbb{S}$,
$\left(\textrm{a}\right)$ is obtained by removing the constants that do not vary with the value of $s_k$,
$q\left(y_n,s_k\right)$ is a function of $y_n$ and $s_k$,
and
\begin{align}\label{equ:q(sk,yn)}
	   &q\left(y_n,s_k\right)\\
	    = &\sum_{\bs_{\setminus k} \in \mathbb{S}^{2K-1}}\exp\Big\{ \sum_{\substack{k'=1, k'\neq k}}^{2K}\left(\bd_{k'}+ \btheta_{n,k'}  \right)^T\bt_{k'} + c_n\left( \bs,y_n \right)  \Big\}.	\nonumber
\end{align}
Note that the proportions in the second line of \eqref{equ:marginals of pn} and  the third line of \eqref{equ:q(sk,yn)2} next will not affect the calculation of $p_{n,k}\left(s_k;\btheta_{n}\right)$ since
the constants corresponding to these proportions do not vary with the value of $s_k$, and thus we can finally normalize $p_{n,k}\left(s_k;\btheta_{n}\right)$.
We will not repeat this property when a similar situation arises in the rest of this paper. 
In the last line of \eqref{equ:marginals of pn}, the calculation of $\exp\braces{\left(\bd_k + \btheta_{n,k}\right)^T\bt_k}$ is simple, 
if we can obtain the approximate value of $q\left(s_k,y_n\right), s_k\in \mathbb{S}$, we then can obtain the approximate value of $p_{n,k}\left(s_k;\btheta_{n}\right), s_k\in \mathbb{S}$.
Hence, our goal now is converted to obtain the approximate value of  $q\left(y_n,s_k\right), s_k\in \mathbb{S}$.
From \eqref{equ:q(sk,yn)}, we can obtain
\begin{align}\label{equ:q(sk,yn)2}
		& \ \ \ \ q\left(y_n,s_k\right)\nonumber \\
		&= \sum_{\bs_{\setminus k} \in \mathbb{S}^{2K-1}}\Big( \prod_{k'=1,k'\neq k}^{2K}\exp\braces{ \left( \bd_{k'}+ \btheta_{n,k'} \right)^T \bt_{k'} }\nonumber\\
		 & \ \ \times\exp\braces{ -\frac{1}{2\sigma_z^2}\left( y_n - \be_n^T\bG\bs \right)^2 }  \Big)\\
		 &\overset{\left(\textrm{a}\right)}{\propto} \sum_{\bs_{\setminus k} \in \mathbb{S}^{2K-1}} \Big( \prod_{k'=1, k'\neq k}^{2K}p_{0,k'}\left( s_{k'};\btheta_{n,k'} \right) f_{\textrm{G}}\left( y_n;\be_n^T\bG\bs,\sigma_z^2 \right) \Big),\nonumber
\end{align}
where 
$\bG$ is defined in \eqref{equ:definition of G},
$\left(\textrm{a}\right)$ is obtained by adding the constant independent with $s_k$ and $y_n$,
$p_{0,k'}\left(s_{k'};\btheta_{n,k'}\right)$ is defined by \eqref{equ:marginals of p0}, 
and $f_{\textrm{G}}\left( x;\mu,\sigma^2 \right)$ denotes the PDF of a real Gaussian distribution $\mathcal{N}\left( \mu,\sigma^2 \right)$ for a real random variable $x$.
Inspired by the last line of \eqref{equ:q(sk,yn)2}, we consider $2N_r\times 2K$ hybrid random variables $Y_{n,k}, n\in \setposi{2N_r}, k \in \setposi{2K}$, where the $\left(n,k\right)$-th of them is defined by: for a given $s_k$, 
\begin{equation}\label{equ:Yn}
	\begin{split}
		Y_{n,k} &= \be_n^T\bG\bs + w
		= g_{n,k}s_k + \sum_{k'=1, k'\neq k}^{2K}g_{n,k'}s_{k'} + w\\
		&= \sum_{k'=1, k'\neq k}^{2K}g_{n,k'}s_{k'} + w'_{n,k},
	\end{split}
\end{equation}
where $s_k$ is considered as a determinate (also known/given) constant,
$g_{n,k}$ is the $\left(n,k\right)$-th component of $\bG$,
$g_{n,k}$ is also considered as a determinate and known constant,
$\braces{s_{k'}}_{k'\neq k}$ are considered as the independent discrete random variables,
the probability distribution of $s_{k'}, k'\neq k$, is given by $p_{0,k'}\left(s_{k'};\btheta_{n,k'}\right)$, 
the joint probability distribution of  $\braces{s_{k'}}_{k'\neq k}$ is then given by $p\left( \bs_{\setminus k} \right) = \prod_{k'\neq k}p_{0,k'}\left(s_{k'};\btheta_{n,k'} \right)$, 
$w \sim \mathcal{N}\left(0,\sigma_z^2\right)$ is a real Gaussian random variable independent with $\braces{s_{k'}}_{k'\neq k}$,
and $w'_{n,k} = w + g_{n,k}s_k \sim \mathcal{N}\left( g_{n,k}s_k,\sigma_z^2 \right)$ is also independent with $\braces{s_{k'}}_{k'\neq k}$.
Briefly, for $Y_{n,k}$   the subscript $n$ determines which row of $\bG$ is multiplied by $\bs$,
and the subscript $k$ determines which component of $\bs$ is considered deterministic.
In this case, it is not difficult to obtain that the PDF of $Y_{n,k}$ is given by \cite[Sec. 6.1.2]{pishro2016introduction} 
\begin{align}\label{equ:PDF of yn}
	   	&f\left(Y_{n,k}\right)\nonumber\\
	   	 = & \sum_{\bs_{\setminus k} \in \mathbb{S}^{2K-1}}\left(\! p\left(\bs_{\setminus k}\right)f_{\textrm{G}}\left(\!Y_{n,k}\!-\!\sum_{k'\neq k}g_{n,k'}s_{k'} ;g_{n,k}s_k,\sigma_z^2  \!\right) \!\right) \nonumber\\
	   	=& \sum_{\bs_{\setminus k} \in \mathbb{S}^{2K-1}}\left(p\left(\bs_{\setminus k}\right)f_{\textrm{G}}\left( Y_{n,k};\be_n^T\bG\bs,\sigma_z^2 \right)  \right),  	
\end{align}
which will be equal to the last line of \eqref{equ:q(sk,yn)2} after we set the value of $Y_{n,k}$ as  $Y_{n,k} = y_n$.
Since although the terms in the summation in \eqref{equ:Yn} are independent each other, they do not have the same distribution. Thus, the conventional CLT may not apply directly. 
We next apply Lyapunov CLT to impose a condition on the values of $g_{n,k}$ in matrix $\bG$ and the variances of the random variables in \eqref{equ:Yn} so that $Y_{n,k}$ converges in distribution to a real Gaussian random variable.
To do so, let us first see Lyapunov CLT. 
\begin{lemma}[Lyapunov central-limit-theorem \cite{billingsley2008probability}]\label{lemma:CLT}
	Suppose $\braces{X_n}_{n=1}^N$ are  independent real random variables, each with finite expected value $\mu_n$ and variance $\sigma_n^2$.
	Denote the random variable $S = \sum_{n=1}^{N}X_n$ and its expected value and variance as $\tilde{\mu} = \sum_{n=1}^{N}\mu_n$ and $\tilde{\sigma}^2 = \sum_{n=1}^{N}\sigma_n^2$, respectively.
	Suppose for some positive $\delta$, Lyapunov's condition 
	\begin{equation}\label{equ:Lyapunov condition}
		\lim\limits_{N\to\infty}\frac{1}{\tilde{\sigma}^{2+\delta}}\sum_{n=1}^N\Exp\braces{ \left| X_n - \mu_n \right|^{2+\delta} } = 0
	\end{equation} 
holds.
Then,  $S$ converges in distribution to a real Gaussian random variable $\tilde{S}$, as $N$ tends to infinity, and 
\begin{equation}
	S\overset{d}{\to}\tilde{S}  \sim \mathcal{N}\left(\tilde{\mu},\tilde{\sigma}^2\right).
\end{equation}
\end{lemma}
Given the probability distribution $p_{0,k'}\left( s_{k'};\btheta_{n,k'} \right)$ of $s_{k'}, k' \in \setposi{2K}\setminus\braces{k}$, in \eqref{equ:Yn}, by using \eqref{equ:probability of marginals of p0} the expected value and the variance of $s_{k'}$ are given by
\begin{subequations}\label{equ:mean and variance of marginals of p0}
	\begin{equation}\label{equ:mu_{n,k}}
		\begin{split}
		\mu_{n,k'} &= \sum_{s_{k'} \in \mathbb{S}}s_{k'}p_{0,k'}\left( s_{k'};\btheta_{n,k'} \right)\\
		&=\frac{ s^{\left(0\right)} + \sum_{\ell=1}^{L-1}s^{\left(\ell\right)} \exp\braces{d_{k',\ell} + \theta_{n,k',\ell}} }{1+\sum_{\ell=1}^{L-1}\exp\braces{ d_{k',\ell} + \theta_{n,k',\ell} }  },
		\end{split}
	\end{equation}
\begin{equation}
	\begin{split}
	   v_{n,k'} &= \sum_{s_{k'} \in \mathbb{S}} s_{k'}^2 p_{0,k'}\left( s_{k'};\btheta_{n,k'} \right) - \mu_{n,k'}^2\\
	   &=\! \frac{ \left(s^{\left(0\right)}\right)^2 + \sum_{\ell=1}^{L-1}\left(s^{\left(\ell\right)}\right)^2 \exp\braces{d_{k',\ell} + \theta_{n,k',\ell}} }{1+\sum_{\ell=1}^{L-1}\exp\braces{ d_{k',\ell} + \theta_{n,k',\ell} }  } \!-\!  \mu_{n,k'}^2.
	\end{split}
\end{equation}
\end{subequations}
Meanwhile, since $\braces{s_{k'}}_{k'\neq k}$ and $w_{n,k}'$ are independent in \eqref{equ:Yn}, the expected value and variance of $Y_{n,k},  n\in \setposi{2N_r}, k\in\setposi{2K}$,  can be readily expressed as
\begin{subequations}
	\begin{equation}
		\Exp\braces{Y_{n,k}} = \sum_{k' = 1, k'\neq k}^{2K}g_{n,k'}\mu_{n,k'} + g_{n,k}s_k,
	\end{equation}
\begin{equation}\label{equ:variance of Y_{n,k}}
	\mathbb{V}\braces{Y_{n,k}} = \sum_{k'=1, k'\neq k}^{2K}g_{n,k'}^2v_{n,k'} + \sigma_z^2,
\end{equation}
\end{subequations}
We then have the following theorem.
\begin{theorem}\label{the:CLT of Y_{n,k}}
	If the following condition
	\begin{equation}\label{equ:condition on v_{n,k}}
		\lim\limits_{K\to \infty}\frac{1}{2K}\sum_{k'=1, k'\neq k}^{2K}g_{n,k'}^2v_{n,k'} = \zeta>0
	\end{equation}
    holds for a positive constant $\zeta$, 
	then $Y_{n,k}$ converges in distribution to a real Gaussian random variable $\tilde{Y}_{n,k}$, as $2K$ goes to infinity, and
	\begin{equation}
			Y_{n,k}\overset{d}{\to}\tilde{Y}_{n,k}  \sim \mathcal{N}\left(\Exp\braces{Y_{n,k}},\mathbb{V}\braces{Y_{n,k}}\right).
	\end{equation}
\end{theorem}
\begin{proof}
	See Appendix \ref{proof:CLT of Y_{n,k}}.
\end{proof}
Intuitively, the condition \eqref{equ:condition on v_{n,k}} means that as $K$ (or, equivalently, $2K-1$) tends to infinity, the variance of the random variable $\tilde{s}_{n,k'} \triangleq g_{n,k'}s_{k'}, k'\in \setposi{2K}\setminus\braces{k}$, in \eqref{equ:Yn} does not tends to zero, where $g_{n,k'}$ is the $\left(n,k'\right)$-th component of $\bG$ defined in \eqref{equ:definition of G},  or $\tilde{s}_{n,k'}$ does not tend to be a deterministic value.
This guarantees that the CLT holds.
When $2K$ is large, from Theorem \ref{the:CLT of Y_{n,k}}, 
$q\left(y_n,s_k\right)$ is approximately proportional to $f_{\textrm{G}}\left(\tilde{Y}_{n,k};\Exp\braces{Y_{n,k}},\mathbb{V}\braces{Y_{n,k}}\right)\big|_{\tilde{Y}_{n,k} = y_n}$, and thus we can obtain
\begin{equation}\label{equ:pn sk app}
	\begin{split}
		&p_{n,k}\left(s_k;\btheta_{n}\right)\\ &\overset{\left(\textrm{a}\right)}{\propto} \exp\braces{\left( \bd_k+\btheta_{n,k} \right)^T\bt_k -\frac{\left( y_n - \Exp\braces{Y_{n,k}} \right)^2}{2\mathbb{V}\braces{Y_{n,k}}}  }\\
		&=  \exp\braces{ \left( \bd_k+\btheta_{n,k} \right)^T\bt_k -\frac{\left( g_{n,k}s_k -  \tilde{\mu}_{n,k}  \right)^2}{2\mathbb{V}\braces{Y_{n,k}}}  },
	\end{split}
\end{equation}
where $s_k \in \mathbb{S}$, $k\in \setposi{2K}$, $n \in  \setposi{2N_r}$,
$\left(\textrm{a}\right)$ is obtained by removing the constant independent with $s_k$ and $y_n$,
and $\tilde{\mu}_{n,k}, n\in \setposi{2N_r}, k\in \setposi{2K}$, is defined as 
\begin{equation}\label{equ:tilde mu_{n,k}}
	\tilde{\mu}_{n,k} \triangleq y_n - \sum\nolimits_{k' = 1, k'\neq k}^{2K}g_{n,k'}\mu_{n,k'},
\end{equation}
As a summary, when $2K$ is large we approximately have
\begin{subequations}
	\begin{equation}
		p_{n,k}\left(s_k;\btheta_{n}\right)\big|_{s_k = s^{\left(0\right)}} = C_{n,k}\exp\braces{  -\frac{\left( g_{n,k}s^{\left(0\right)} -  \tilde{\mu}_{n,k}  \right)^2}{2\mathbb{V}\braces{Y_{n,k}}}  },
	\end{equation}
	\begin{equation}\label{equ:pro of pn ell}
		\begin{split}
		&p_{n,k}\left(s_k;\btheta_{n}\right)\big|_{s_k = s^{\left(\ell\right)}}\\
		= &C_{n,k}\exp\braces{d_{k,\ell} + \theta_{n,k,\ell}  -\frac{\left( g_{n,k}s^{\left(\ell\right)} -  \tilde{\mu}_{n,k}  \right)^2}{2\mathbb{V}\braces{Y_{n,k}}}  },	
		\end{split}
	\end{equation}
\end{subequations}
where $C_{n,k}$ is the normalization factor,
and $\ell \in \setposi{L-1}$ in \eqref{equ:pro of pn ell}.
Combining \eqref{equ:mp solution}, we can immediately obtain that
\begin{equation}\label{equ:mp solution2}
	\begin{split}
		\theta_{0n,k,\ell} = &\frac{g_{n,k}\left( s^{\left(0\right)}-s^{\left(\ell\right)} \right)\left[ g_{n,k}\left(s^{\left(0\right)} + s^{\left(\ell\right)}\right)-2\tilde{\mu}_{n,k} \right]}{2\mathbb{V}\braces{Y_{n,k}}}\\
		&+\theta_{n,k,\ell}, 
	\end{split}
\end{equation}
where $\ell \in \setposi{L-1}, k\in \setposi{2K}$, and $n\in \setposi{2N_r}$.
Hence, we obtain an approximate solution of the $m$-projection $p_0\left(\bs;\btheta_{0n}\right), n\in \setposi{2N_r}$.
From Theorem \ref{the:CLT of Y_{n,k}}, it is not difficult to check that when the condition \eqref{equ:condition on v_{n,k}} holds, \eqref{equ:mp solution2} is asymptotically accurate as $K$ goes infinity.
From \eqref{equ:xi_n in preli}, the approximation term $\bxi_{n} = \btheta_{0n} - \btheta_{n}$ can be then expressed as
\begin{subequations}\label{equ:xi_{n,k}}
	\begin{equation}
		\bxi_{n} = \left[ \bxi_{n,1}^T, \bxi_{n,2}^T, \ldots, \bxi_{n,2K}^T \right]^T
	\end{equation}
	\begin{equation}
		\bxi_{n,k} = \left[ \xi_{n,k,1}, \xi_{n,k,2},\ldots,\xi_{n,k,L-1}   \right]^T,
	\end{equation}
	\begin{equation}\label{equ:xi_n,k,ell}
		\xi_{n,k,\ell} 
		=\frac{g_{n,k}\left( s^{\left(0\right)}-s^{\left(\ell\right)} \right)\left[ g_{n,k}\left(s^{\left(0\right)} + s^{\left(\ell\right)}\right)-2\tilde{\mu}_{n,k} \right]}{2\mathbb{V}\braces{Y_{n,k}}},
	\end{equation}
\end{subequations}
where $n\in \setposi{2N_r}$, $k\in \setposi{2K}$, and $\ell \in \setposi{L-1}$.
We give the detailed expression of $\xi_{n,k,\ell}$ in \eqref{equ:xi_{n,k,l}}, where $g_{n,k}$ is the $\left(n,k\right)$-th component of the real-valued channel matrix $\bG$ in \eqref{equ:rece model}, $\braces{s^{\left(\ell\right)}}_{\ell = 0}^{L-1}$ defined below \eqref{equ:rece model} are the constellation points for the components of $\bs$, $y_n$ is the $n$-th component of the received signal $\by$ in \eqref{equ:rece model}, $d_{k,\ell}$ is the NP defined by \eqref{equ:d_{k,l}}, and $\sigma_z^2$ is the noise variance of $\bz$ in \eqref{equ:rece model}.

\begin{figure*}[!h]
	\begin{subequations}\label{equ:xi_{n,k,l}}
		\begin{equation}
			\xi_{n,k,\ell} 
			=\frac{g_{n,k}\left( s^{\left(0\right)}-s^{\left(\ell\right)} \right)\braces{ g_{n,k}\left(s^{\left(0\right)} + s^{\left(\ell\right)}\right)-2\left[ y_n - \sum\limits_{k' = 1, k'\neq k}^{2K}g_{n,k'}\left( \frac{ s^{\left(0\right)} + \sum\limits_{\ell=1}^{L-1}s^{\left(\ell\right)} \exp\braces{d_{k',\ell} + \theta_{n,k',\ell}} }{1+\sum\limits_{\ell=1}^{L-1}\exp\braces{ d_{k',\ell} + \theta_{n,k',\ell} }  } \right) \right]  }}{2\mathbb{V}\braces{Y_{n,k}}}
		\end{equation}
		\begin{equation}
			\mathbb{V}\braces{Y_{n,k}} = \sum_{k'=1, k'\neq k}^{2K}g_{n,k'}^2\left[ \frac{ \left(s^{\left(0\right)}\right)^2 + \sum\limits_{\ell=1}^{L-1}\left(s^{\left(\ell\right)}\right)^2 \exp\braces{d_{k',\ell} + \theta_{n,k',\ell}} }{1+\sum\limits_{\ell=1}^{L-1}\exp\braces{ d_{k',\ell} + \theta_{n,k',\ell} }  } \!-\!  \left( \frac{ s^{\left(0\right)} + \sum\limits_{\ell=1}^{L-1}s^{\left(\ell\right)} \exp\braces{d_{k',\ell} + \theta_{n,k',\ell}} }{1+\sum\limits_{\ell=1}^{L-1}\exp\braces{ d_{k',\ell} + \theta_{n,k',\ell} }  }  \right)^2 \right] + \sigma_z^2
		\end{equation}
	\end{subequations}
	\hrule
\end{figure*}

After the approximate $\bxi_n$ is obtained, we update the EACSs of $p_n\left(\bs;\btheta_{n}\right), n\in \setnnega{2N_r}$, as 
\begin{subequations}\label{equ:update of NPs IGA}
	\begin{equation}
		\btheta_{n}\left(t+1\right) = \alpha \sum_{n' = 1, n'\neq n}^{2N_r}\bxi_{n'}\left(t\right) + \left(1-\alpha\right)\btheta_{n}\left(t\right), n\in \setposi{2N_r},
	\end{equation}
	\begin{equation}
		\btheta_{0}\left(t+1\right) = \alpha\sum_{n=1}^{2N_r}\bxi_n\left(t\right)+\left(1-\alpha\right)\btheta_{0}\left(t\right),
	\end{equation}
\end{subequations}
where $0<\alpha\le 1$ is the damping,
and repeat the $m$-projections, calculating  $\bxi_n$ and updating until convergence.
We summarize the IGA-SD in Algorithm \ref{Alg:IGA}.
The computational complexity (the number of real-valued multiplications) of the IGA-SD is $\mathcal{O}\left(16N_rK\left( L+1 \right)   \right)$ (the number of real-valued multiplications) per iteration, where $N_r$ is the number of antennas at the BS , $K$ is the number of users, $L = \sqrt{\tilde{L}}$, and $\tilde{L}$ is the modulation order.

%

\section{Simulation Results}
In this section, we provide simulation results to illustrate
the performance of the proposed  IGA-SD. 
The uncoded bit error rate (BER) is adopted as the performance metric. 
In our simulations, the BS comprises a uniform planar array (UPA) of $N_r = N_{r,v}\times N_{r,h}$ antennas, and $N_{r,v}$ and $N_{r,h}$ are the numbers of the antennas at each vertical column and horizontal row, respectively.
We average our results for $1000$ realizations of the channel matrix $\bG$, which is generated by the widely adopted QuaDRiGa \cite{quadriga}. 
We set the simulation scenario to "3GPP\_38.901\_UMa\_NLOS", and the main parameters for the simulations are summarized in Table \ref{tab:para}. 
\begin{table}[htbp]
	\centering
	\caption{Parameter Settings of the Simulation}\label{tab:para}
	\begin{tabular}{cc}
		\hline
		Parameter &Value \\
		\hline
		Number of BS antennas $N_{r,v}\times N_{r,h}$ & $16\times 64$ \\
		UT number $K$ & $240$ \\
		Center frequency $f_c$ & $4.8$GHz \\
		Modulation Mode & QAM \\
		Modulation Order $\tilde{L}$ & $4$, $16$, and $64$ \\
		\hline
	\end{tabular}
\end{table}
The BS  is located at $\left( 0,0,25 \right)$.
The users are randomly generated in a $120^{\circ}$ sector with radius $r = 200$m around $(0, 0, 1.5)$. 
The channel matrix is normalized as $\Exp\left\lbrace  \lVert \bG \rVert_F^2 \right\rbrace = N_rK $. 
The average power of the transmitted symbol of each user is normalized to $1$, and the $\textrm{SNR}$ is set as $\textrm{SNR} = \frac{K}{\tilde{\sigma}_z^2}$. 
Based on the received signal model \eqref{equ:rece model},
we compare the proposed IGA-SD with the following detectors.\\
\textbf{LMMSE}: The linear minimum-mean-squared error (LMMSE) detector with hard-decision. 
The LMMSE detector is given by
\begin{equation}
	\bs_{\textrm{MMSE}} = \left( \bG^T\bG + \sigma_z^2\bI \right)^{-1}\bG^T\by.
\end{equation}
Then, a component-wise hard-decision is performed as
\begin{equation}
	s_{k,\textrm{MMSE}} = \argmin{s_k \in \mathbb{S}}\left| s_k - \left[ \bs_{\textrm{MMSE}} \right]_k \right|^2, k\in \setposi{2K}.
\end{equation} 
\\
\textbf{EP}: The expectation propagation detector proposed in \cite{6841617}, where the hard-decision is also performed.\\
\textbf{AMP}: Approximate message passing algorithm proposed in \cite{AMP}.
AMP can obtain the approximations of the marginals of the \textsl{a posteriori} distribution $p\left(\bs|\by\right)$.
Thus, AMP is used as an MPM detector (\eqref{equ:MPM dete}).

The computational complexity of the LMMSE detector is $\mathcal{O}\left(8\left(2N_rK^2 + K^3 \right) \right)$ \cite{6841617}.
The computational complexity of the EP detector and AMP are $\mathcal{O}\left( 8\left(N_rK^2 + K^3\right) \right)$ and $\mathcal{O}\left( 8\left(N_rK \right) \right)$ per iteration, respectively \cite{6841617,AMP}.
The complexity of EP detector is the highest among all algorithms. 
When the number of iterations is low (e.g., tens), the complexity of IGA-SD is lower than that of LMMSE detection. 
The computational complexity of AMP is the lowest.

\begin{figure}[htbp]
	\centering
	\includegraphics[width=0.43\textwidth]{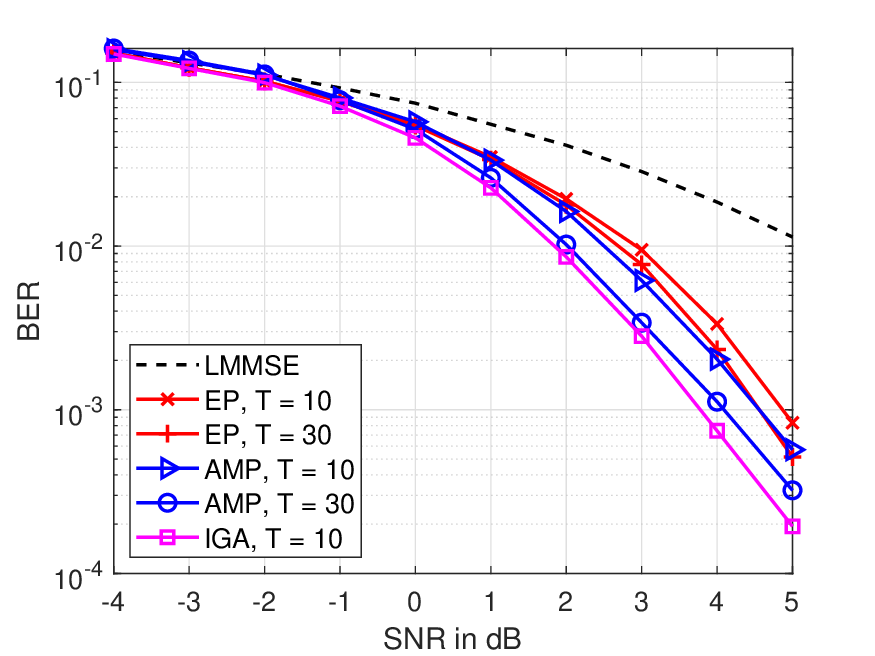}
	\caption{\small BER performance of IGA compared with AMP, EP and LMMSE under $4$-QAM.}
	\label{fig:SNR_4QAM}
\end{figure}

\begin{figure}[htbp]
	\centering
		\includegraphics[width=0.43\textwidth]{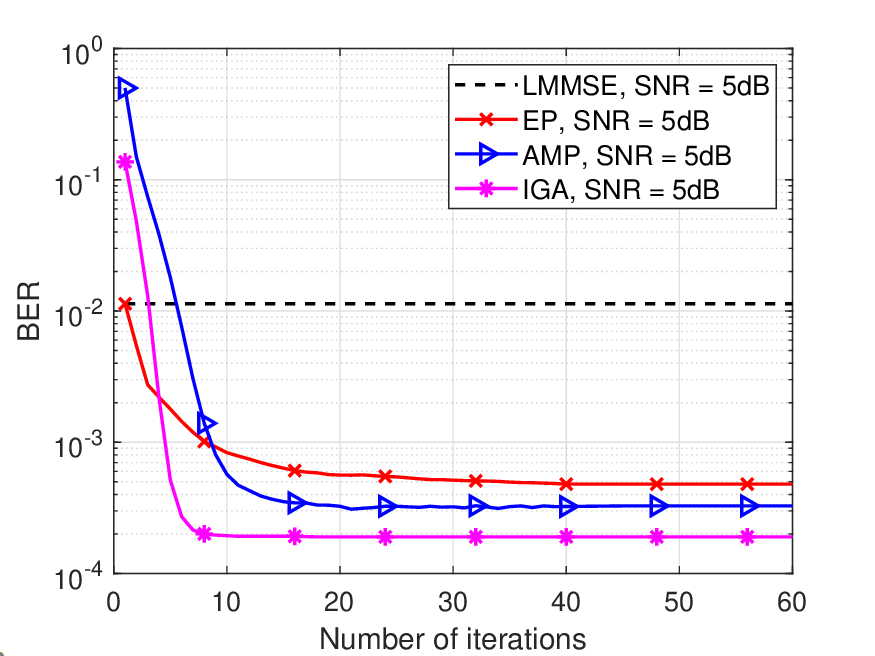}
			\caption{\small Convergence performance of IGA compared with EP and AMP at SNR = $5$ dB under $4$-QAM.}
		\label{fig:It_4QAM}
\end{figure}

We first consider $4$-QAM modulation.
Fig. \ref{fig:SNR_4QAM} shows the BER performance of the IGA-SD compared with LMMSE, EP and AMP.
The iteration numbers of IGA-SD, EP and AMP are set as $10$, $10$ and $30$, and $10$ and $30$, respectively.
Meanwhile, the convergence performance of the iterative algorithms at SNR = $5$dB is shown in Fig. \ref{fig:It_4QAM}.
From Fig. \ref{fig:SNR_4QAM}, we can find that all the iterative algorithms outperform the LMMSE detector within limited iteration numbers.
For BER = $10^{-3}$, the SNR gains of the IGA-SD with $10$ iterations compared to the AMP with $10$ and $30$ iterations are around $0.7$dB and $0.3$dB, respectively.
Meanwhile, IGA-SD with $10$ iterations can improve the EP performance with $10$ and $30$ iterations in $1$dB and $0.7$dB for BER = $10^{-3}$, respectively.
From Fig. \ref{fig:It_4QAM}, it can be found that in the case with $4$-QAM and SNR = $5$dB, the IGA-SD requires around $10$ iterations to converge and achieves the lowest BER performance.
AMP and EP require about $25$ and $45$ iterations to converge, respectively.
The decrease in BER is minor after $30$ iterations for EP.
Moreover, we can find that the BER performance of EP with one iteration is equal to that of LMMSE detector.
This can be attributed to that the EP detector with one iteration is equivalent to the LMMSE detector \cite{6841617}.

\begin{figure}[htbp]
		\centering
			\includegraphics[width=0.43\textwidth]{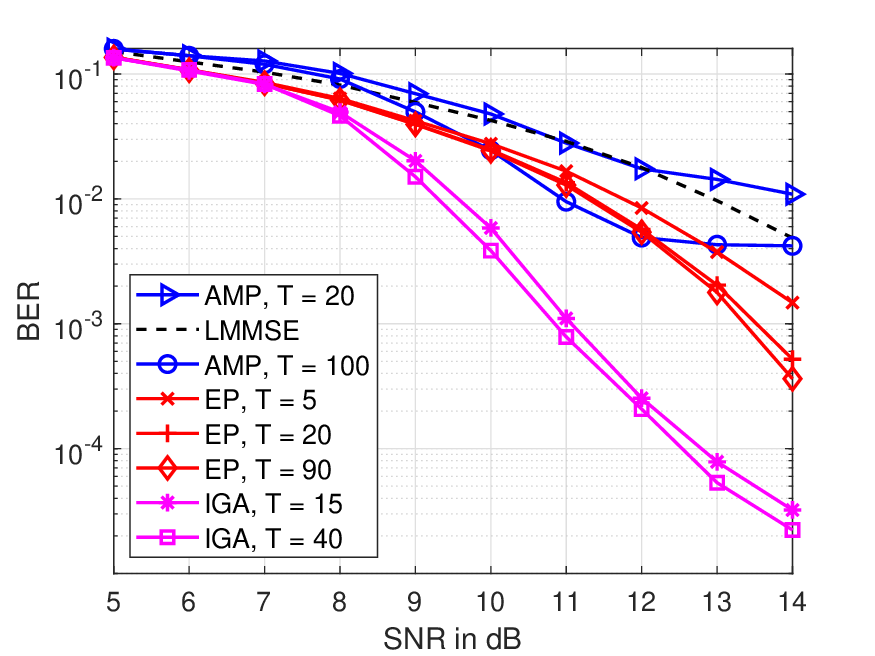}
			\caption{\small BER performance of IGA compared with AMP, EP and LMMSE under $16$-QAM.}
			\label{fig:SNR_16QAM}
\end{figure}

\begin{figure}[htbp]
	\centering
	\includegraphics[width=0.43\textwidth]{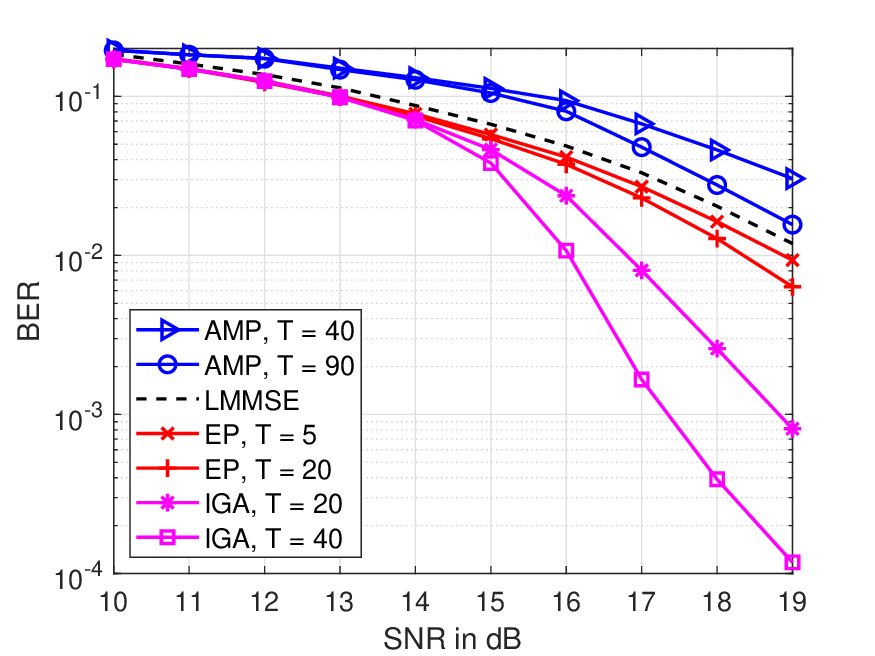}
	\caption{\small BER performance of IGA compared with AMP, EP and LMMSE under $64$-QAM.}
	\label{fig:SNR_64QAM}
\end{figure}

\begin{figure}[htbp]
	\centering
			\includegraphics[width=0.43\textwidth]{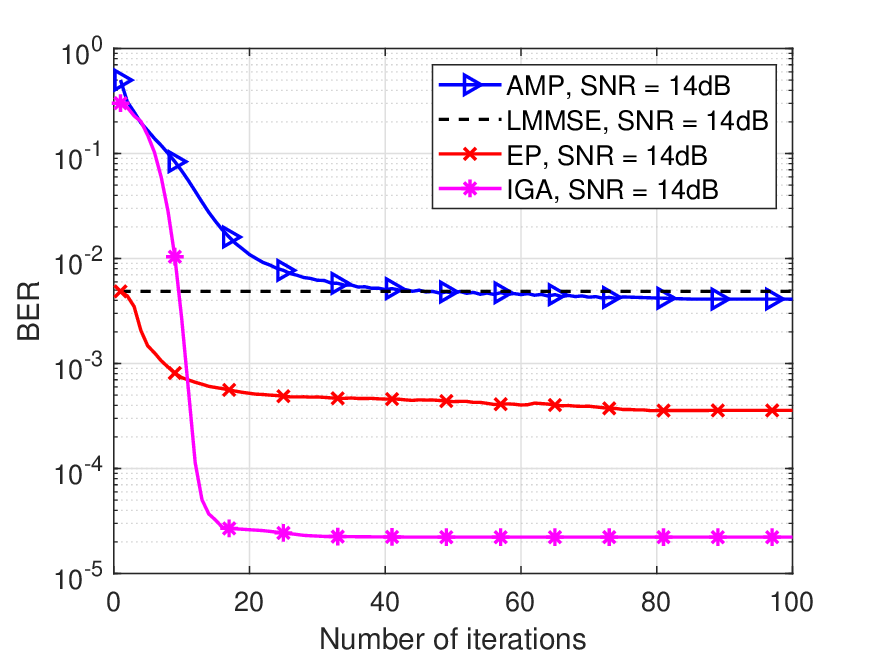}
			\caption{\small Convergence performance of IGA compared with EP and AMP at SNR = $14$ dB under $16$-QAM.}
			\label{fig:It_16QAM}
\end{figure}

\begin{figure}[htbp]
			\centering
			\includegraphics[width=0.43\textwidth]{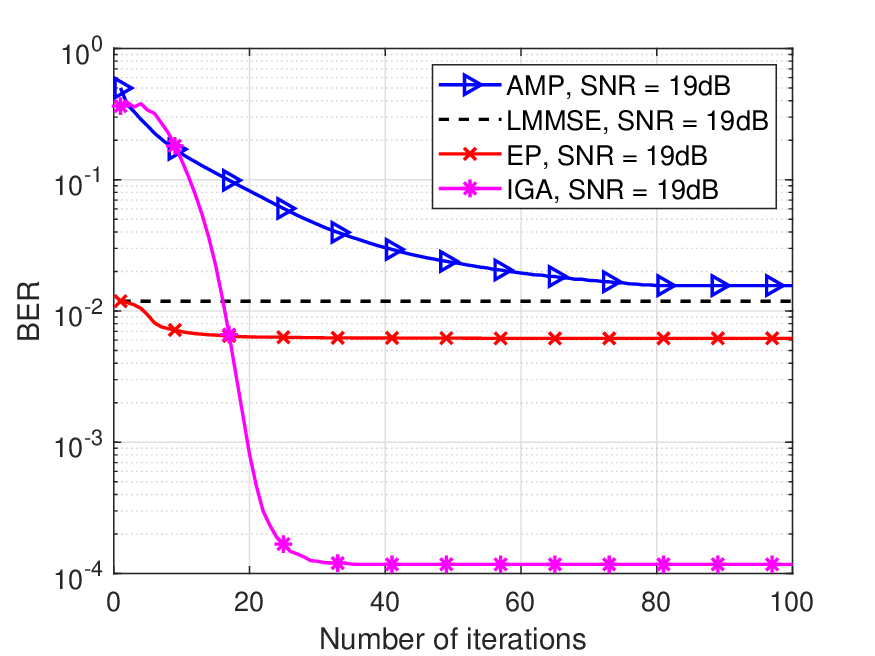}
			\caption{\small Convergence performance of IGA compared with EP and AMP at SNR = $19$ dB under $64$-QAM.}
			\label{fig:It_64QAM}
\end{figure}

Fig. \ref{fig:SNR_16QAM} and and \ref{fig:SNR_64QAM} show the BER performance for $16$-QAM and $64$-QAM, respectively.
From Fig. \ref{fig:SNR_16QAM}, we can find that the BER performance of LMMSE outperforms that of the AMP with $20$ iterations.
Meanwhile, we can find that the gap between IGA-SD and the other algorithms is increasing.
For BER = $10^{-3}$, the SNR gains of the IGA-SD with $15$ iterations compared to the EP with $20$ and $90$ are about $1.2$dB and $0.9$dB, respectively.
The SNR gain for the the IGA-SD with $40$ iterations increases by about $0.2$dB each over the two gains above.
For $64$-QAM, from Fig. \ref{fig:SNR_64QAM}, we can find that the BER performance of the LMMSE detector exceeds  that of the AMP after convergence.  
The gap between IGA-SD and the other algorithms is still increasing.
For BER = $10^{-2}$, IGA-SD with $20$ iterations has improved the EP performance with $5$ and $20$ iterations in $2.1$dB and $1.6$dB, respectively.
The SNR gain for the the IGA-SD with $40$ iterations increases by about $0.7$dB each over the two gains above.

We then show the convergence performances for the cases with $16$-QAM and $64$-QAM in Fig. \ref{fig:It_16QAM} and \ref{fig:It_64QAM}, respectively.
From Fig. \ref{fig:It_16QAM}, it can be found that in the case with $16$-QAM and SNR = $14$dB, the IGA-SD requires around $30$ iterations to converge and achieves the best BER performance.
Both AMP and EP require about $90$ iterations to converge.
From Fig. \ref{fig:It_64QAM}, we can find that in the case with $64$-QAM and SNR = $19$dB, the IGA-SD requires around $30$ iterations to converge and achieves the best BER performance.
AMP and EP require about $80$ and $20$ iterations to converge, respectively.
We can also find that compared to $16$-QAM, AMP and EP require fewer iterations to converge at $64$-QAM modulation.
This could be attributed to the fact that compared to $16$-QAM, the converged BER performances of both AMP and EP have severely degraded in the case with $64$-QAM.

\section{Conclusion}
We have proposed an information geometry approach for ultra-MIMO signal detection in this paper.
The signal detection is formulated as an MPM detection problem based on the approximation of the \textsl{a posteriori} marginals of the transmitted symbols of all users.
To obtain the approximation of the \textsl{a posteriori} marginals, the information geometry theory is introduced.
Specifically, we convert the calculation of the approximation of the \textsl{a posteriori} marginals into an iterative $m$-projection process.
Then, the Lyapunov CLT is applied to have an approximate solution of the $m$-projection between a probability distribution of the AM and the OBM.
Simulation results verify that the IGA-SD can obtain the lowest BER performance within a limited number of iterations compared with the existing approaches, which demonstrates the superiority of the proposed IGA-SD for ultra-massive MIMO systems.

\appendices

\section{Proof of Theorem \ref{lemma:equivalent of mp1}}\label{proof:equivalent of mp1}
We prove this theorem by showing that problem \eqref{equ:mp in detection 1} is strictly convex w.r.t. $\btheta_{0}$.
More precisely, we show that the Hessian of the objective function  defined by the K-L divergence \eqref{equ:K-L divergence in detection} is a positive definite matrix. 
Before proceeding, we first introduce the expectation parameter (EP) and the Fisher information matrices (FIM) of
$p_0\left(\bs;\btheta_0\right)$ and $p_n\left(\bs;\btheta_n\right), n\in \setposi{2N_r}$ in \eqref{equ:defintion of p_0} and \eqref{equ:pn}, respectively.
$\braces{p_n\left(\bs;\btheta_n\right)}_{n=0}^{2N_r}$ can be expressed as
\begin{equation}
	p_n\left(\bs;\btheta_n\right) = \exp\braces{\btheta_n^T\bt + m_n\left(\bs\right) - \psi_n\left(\btheta_n\right)},
\end{equation} 
where $n\in \setnnega{2N_r}$ and $m_n\left(\bs\right)$ is a function independent of $\btheta_n$.
Specifically, we have $m_0\left(\bs\right) \triangleq \bd^T\bt$ for $p_0\left(\bs;\btheta_0\right)$,
and  $m_n\left(\bs\right) = \bd^T\bt + c_n\left(\bs,y_n\right)$ for $p_n\left(\bs;\btheta_n\right), n\in\setposi{2N_r}$.
Since the free energy $\psi_n\left(\btheta_n\right)$ is constrained by the normalization condition, we have
\begin{equation}\label{equ:constrain of free energy1}
	1 = \sum_{\bs \in \mathbb{S}^{2K}} \exp\braces{\btheta_n^T\bt + m_n\left(\bs\right) - \psi_n\left(\btheta_n\right)}, n\in \setnnega{2N_r},
\end{equation}
from which we can obtain
\begin{equation}\label{equ:constrain of free energy2}
	\psi_n\left(\btheta_n\right) = \ln \left( \sum_{\bs \in \mathbb{S}^{2K}}\exp\braces{ \btheta_n^T\bt + m_n\left(\bs\right) } \right), n\in \setnnega{2N_r}.
\end{equation}
Then, from \eqref{equ:constrain of free energy2}, the partial derivative of $\psi_n\left(\btheta_n\right)$ is 
\begin{align}
		&\frac{\partial \psi_n\left(\btheta_n\right)}{\partial \btheta_n } \nonumber \\
		=&\ \frac{1}{ \sum_{\bs \in \mathbb{S}^{2K}}\exp\braces{ \btheta_n^T\bt + m_n\left(\bs\right) } } 
		 \sum_{\bs \in \mathbb{S}^{2K}}\exp\braces{\btheta_n^T\bt + m_n\left(\bs\right)  }\bt \nonumber \\
	  \equaa& \ \exp\braces{-\psi_n\left(\btheta_n\right)}\sum_{\bs \in \mathbb{S}^{2K}}\exp\braces{\btheta_n^T\bt + m_n\left(\bs\right)  }\bt\nonumber \\
	  =& \ \sum_{\bs \in \mathbb{S}^{2K}}p_n\left(\bs;\btheta_n\right)\bt = \Exp_{p_n\left(\bs;\btheta_n\right)}\braces{ \bt } = \bbeta_n\left(\btheta_n\right),
\end{align}
where $\left(\textrm{a}\right)$ comes from \eqref{equ:constrain of free energy1}.
$\bbeta_n\left(\btheta_n\right) \in \bbR^{2K\left(L-1\right)}$ above is referred as to the EP of $p_n\left(\btheta_n\right), n\in \setnnega{2N_r}$.
Then, the Hessian of $\psi_n\left( \btheta_n \right)$ is 
\begin{align}
	\boldsymbol{\mathcal{I}}_n\left(\btheta_n\right) &\triangleq \frac{\partial^2 \psi_n\left(\btheta_n\right)}{\partial \btheta_n \partial \btheta_n^T}=\frac{\partial\bbeta_n\left(\btheta_n\right)}{\partial \btheta_n^T} 
	= \sum_{\bs \in \mathbb{S}^{2K}}\bt\frac{\partial p_n\left( \bs;\btheta_n \right)}{\partial \btheta_n^T}\nonumber \\
	&= \sum_{\bs \in \mathbb{S}^{2K}}\bt p_n\left( \bs;\btheta_n \right)\left( \bt^T - \frac{\partial \psi_n\left(\btheta_n\right)}{\partial \btheta_n^T } \right) \nonumber \\
	&=\sum_{\bs \in \mathbb{S}^{2K}}\bt p_n\left( \bs;\btheta_n \right)\left( \bt - \bbeta_n\left(\btheta_n\right)  \right)^T\\
	&= \sum_{\bs \in \mathbb{S}^{2K}} p_n\left( \bs;\btheta_n \right)\left( \bt - \bbeta_n\left(\btheta_n\right)  \right)\left( \bt - \bbeta_n\left(\btheta_n\right)  \right)^T
	 \nonumber\\
	 &= \Exp_{p_n\left(\bs;\btheta_n\right)}\braces{ \left( \bt - \bbeta_n\left(\btheta_n\right)  \right)\left( \bt - \bbeta_n\left(\btheta_n\right)  \right)^T} , \nonumber 
\end{align}
where $n \in \setnnega{2N_r}$.
$\boldsymbol{\mathcal{I}}_n\left(\btheta_n\right) \in \bbR^{2K\left(L-1\right)\times 2K\left(L-1\right)}$ above is referred as to the FIM of $p_n\left(\bs;\btheta_n\right), n\in \setnnega{2N_r}$.
From the definition, we can readily show that $\boldsymbol{\mathcal{I}}_n\left(\btheta_n\right), n\in \setnnega{2N_r}$, is positive semi-definite.
Particularly, the FIM of $p_0\left(\bs;\btheta_{0}\right)$ in \eqref{equ:defintion of p_0} is positive definite.
The reason is as follows.
Since $\braces{s_k}_{k=1}^{2K}$ are independent with each other given the joint probability distribution $p_0\left(\bs;\btheta_{0}\right)$, 
we can readily obtain $	p_0\left(\bs;\btheta_{0}\right) = \prod_{k=1}^{2K}p_{0,k}\left(s_k;\btheta_{0,k}\right),$
where $p_{0,k}\left(s_k;\btheta_{0,k}\right)$ defined by \eqref{equ:marginals of p0} is the probability distribution of $s_k, k\in \setposi{2K}$.
Then, $\bbeta_0\left(\btheta_{0}\right)\in \bbR^{2K\left(L-1\right)}$ can be expressed as 
\begin{equation}
	\bbeta_0\left(\btheta_{0}\right) = \left[\bbeta_{0,1}^T\left(\btheta_{0,1}\right), \bbeta_{0,2}^T\left(\btheta_{0,2}\right),\ldots,\bbeta_{0,2K}^T\left(\btheta_{0,2K}\right) \right]^T, 
\end{equation}
where $\bbeta_{0,k}\left( \btheta_{0,k} \right) \in \bbR^{\left(L-1\right)}$ is given by,
\begin{align}\label{equ:EP of marginals of p0}
	\bbeta_{0,k}\left( \btheta_{0,k} \right) &= \Exp_{p_0\left(\bs;\btheta_{0}\right)}\braces{\bt_k} = \Exp_{p_{0,k}\left(s_k;\btheta_{0,k}\right)}\braces{\bt_k}. 
\end{align}
From the last equation in \eqref{equ:EP of marginals of p0}, we refer to $\bbeta_{0,k}\left( \btheta_{0,k} \right)$ as the EP of $p_{0,k}\left(s_k;\btheta_{0,k}\right)$.
Denote the $\ell$-th component in $\bbeta_{0,k}\left( \btheta_{0,k} \right)$ as $\eta_{0,k,\ell}\left( \btheta_{0,k} \right), \ell \in \setposi{L-1}, k\in \setposi{2K}$.
From  $\bt_k =  \left[t_{k,1},t_{k,2},\ldots,t_{k,L-1}  \right]^T , k\in \setposi{2K}$, 
and $t_{k,\ell} = \delta\left( s_k - s^{\left(\ell\right)} \right), \ell \in \setposi{L-1}$,
we can obtain
\begin{align}\label{equ:eta 0,k,l}
	\eta_{0,k,\ell}\left( \btheta_{0,k} \right) &= \Exp_{p_{0,k}\left(s_k;\btheta_{0,k}\right)}\braces{\delta\left( s_k - s^{\left(\ell\right)} \right)} \nonumber\\
	&= p_{0,k}\left(s_k;\btheta_{0,k}\right)\Big|_{s_k = s^{\left(\ell\right)}} > 0.
\end{align}
We then define $\left(L-1\right)\times \left(L-1\right)$ dimensional covariance matrices 
$\bC\left( \bt_k,\bt_{k'} \right)$ as
\begin{align}
	&\bC\left( \bt_k,\bt_{k'} \right) \triangleq \nonumber\\ 
	&\Exp_{p_0\left(\bs;\btheta_{0}\right)}\braces{ \left( \bt_k-\bbeta_{0,k}\left(\btheta_{0,k}\right) \right)   \left( \bt_{k'}-\bbeta_{0,k'}\left(\btheta_{0,k'}\right) \right)^T },
\end{align}
where $k, k' \in \setposi{2K}$.
Since $\braces{s_k}_{k=1}^{2K}$ are independent with each other given $p_0\left(\bs;\btheta_{0}\right)$, we can obtain 
\begin{align}\label{equ:Covariance of tk and tk'}
	&\bC\left(\bt_k,\bt_{k'}\right)   \\
  =	&\begin{cases}
  \bR\left(\bt_k\right)- \bbeta_{0,k}\left(\btheta_{0,k}\right)\bbeta^T_{0,k}\left(\btheta_{0,k}\right), & \textrm{when} \ k = k',\\
  		\mathbf{0}, & \textrm{otherwise},
  	\end{cases}\nonumber
\end{align}
where 
$\mathbf{0} \in \bbR^{\left(L-1\right)\times \left(L-1\right)}$ is the zero matrix,
$\bR\left(\bt_k\right) \in \bbR^{\left(L-1\right)\times \left(L-1\right)}$ is given by
\begin{equation}
	 \bR\left(\bt_k\right) = \Exp_{p_0\left(\bs;\btheta_{0}\right)}\braces{ \bt_k\bt_k^T }=	\Exp_{p_{0,k}\left(s_k;\btheta_{0,k}\right)}\braces{ \bt_k\bt_k^T }.
\end{equation}
From the definition of $\bt_k$, the $\left(i,j\right)$-th element in $\bR\left(\bt_k\right)$ can be expressed
as
\begin{align}
	&\left[ \bR\left(\bt_k\right) \right]_{i,j} = \Exp_{p_{0,k}\left(s_k;\btheta_{0,k}\right)}\braces{ \delta\left( s_k - s^{\left(i\right)} \right) \delta\left( s_k - s^{\left(j\right)} \right) } \nonumber \\
	&=
	\begin{cases}
		p_{0,k}\left(s_k;\btheta_{0,k}\right)\Big|_{s_k = s^{\left(i\right)}} = \eta_{0,k,i}\left(\btheta_{0,k}\right), &\textrm{when}\ i=j,\\
		0, &\textrm{otherwise},
	\end{cases}
\end{align}
where $i,j \in \setposi{L-1}$.
Hence, we can obtain $\bR\left(\bt_k\right) = \Diag{\bbeta_{0,k}\left(\btheta_{0,k}\right)}, k \in \setposi{2K}$.
From \eqref{equ:eta 0,k,l}, we can readily check that $\bR\left(\bt_k\right)$ is positive definite.
Also, we refer to $\bC\left(\bt_k,\bt_k\right)$ as the FIM of $p_{0,k}\left(s_k;\btheta_{0,k}\right), k\in \setposi{2K}$, and we denote  $\boldsymbol{\mathcal{I}}_{0,k}\left(\btheta_{0,k}\right) \triangleq \bC\left(\bt_k,\bt_k\right)$.
The FIM of $p_0\left(\bs;\btheta_{0}\right)$ can be then  expressed as
\begin{align}\label{equ:FIM of p0}
	&\boldsymbol{  \mathcal{I}}_0\left(\btheta_{0}\right) = \nonumber \\
	&\left[ {\begin{array}{*{20}{c}}
			{\boldsymbol{\mathcal{I}}_{0,1}\left(\btheta_{0,1}\right)} & {\bC\left(\bt_1,\bt_2\right) } & {\cdots} & { \bC\left(\bt_1,\bt_{2K}\right)}\\
			{ \bC\left(\bt_2,\bt_1\right) }&{ \boldsymbol{\mathcal{I}}_{0,2}\left(\btheta_{0,2}\right) }&{ \cdots }&{ \bC\left(\bt_2,\bt_{2K}\right) }\\
			{ \vdots }&{ \cdots }&{ \ddots }&{ \vdots }\\
			{ \bC\left(\bt_{2K},\bt_1\right) }&{ \cdots }&{ \cdots }&{ \boldsymbol{\mathcal{I}}_{0,2K}\left(\btheta_{0,2K}\right) }
	\end{array}} \right]. 
\end{align}
From \eqref{equ:Covariance of tk and tk'}, we can obtain that the FIM $\boldsymbol{  \mathcal{I}}_0\left(\btheta_{0}\right)$ of $p_0\left(\bs;\btheta_{0}\right)$ is a block diagonal matrix with the FIMs of $\braces{p_{0,k}\left(s_k;\btheta_{0,k}\right)}_{k=1}^{2K}$ located along its main diagonal, 
i.e.,
\begin{align}\label{equ:FIM of p02}
	& \ \boldsymbol{\mathcal{I}}_0\left(\btheta_{0}\right)  \\
	= &\ \Bdiag{\boldsymbol{  \mathcal{I}}_{0,1}\left(\btheta_{0,1}\right),\boldsymbol{  \mathcal{I}}_{0,2}\left(\btheta_{0,2}\right), \ldots,\boldsymbol{  \mathcal{I}}_{0,2K}\left(\btheta_{0,2K}\right) }.\nonumber
\end{align}
We then show that each FIM $\boldsymbol{\mathcal{I}}_{0,k}\left(\btheta_{0,k}\right), k\in \setposi{2K}$, is positive definite.
From the definition, we have
\begin{align}
	&\boldsymbol{\mathcal{I}}_{0,k}\left(\btheta_{0,k}\right) = \bR\left(\bt_k\right) - \bbeta_{0,k}\left(\btheta_{0,k}\right)\bbeta^T_{0,k}\left(\btheta_{0,k}\right) \nonumber \\
	=&\Diag{\bbeta_{0,k}\left(\btheta_{0,k}\right)} - \bbeta_{0,k}\left(\btheta_{0,k}\right)\bbeta^T_{0,k}\left(\btheta_{0,k}\right).
\end{align}
Given a non-zero vector $\ba  = \left[ a_1,a_2,\ldots,a_{L-1} \right]\in \bbR^{\left(L-1\right)}$, we have (abbreviate $\eta_{0,k,\ell}\left( \btheta_{0,k} \right)$ to $\eta_{0,k,\ell}$ starting from the second equation)
\begin{align}
	&\ba^T\boldsymbol{\mathcal{I}}_{0,k}\left(\btheta_{0,k}\right)\ba \nonumber\\
	=& \sum_{\ell=1}^{L-1}\eta_{0,k,\ell}\left( \btheta_{0,k} \right)a_\ell^2 - \left(\sum_{\ell=1}^{L-1}\eta_{0,k,\ell}\left( \btheta_{0,k} \right)a_\ell  \right)^2 \nonumber \\
	=&\sum_{\ell=1}^{L-1}\eta_{0,k,\ell}a_\ell^2 - \left( \sum_{\ell=1}^{L-1}\sqrt{\eta_{0,k,\ell}}\sqrt{\eta_{0,k,\ell} a_\ell^2}  \right)^2 \nonumber \\
	{\ge}&\sum_{\ell=1}^{L-1}\eta_{0,k,\ell}a_\ell^2 - \left(\sum_{\ell=1}^{L-1}\eta_{0,k,\ell}\right) \left( \sum_{\ell=1}^{L-1}\eta_{0,k,\ell}a_\ell^2 \right) \nonumber \\
	=&\left( 1- \sum_{\ell=1}^{L-1}\eta_{0,k,\ell}\right)\sum_{\ell=1}^{L-1}\eta_{0,k,\ell}a_\ell^2 \nonumber \\
	\overset{\left(\textrm{a}\right)}{=}& \ p_{0,k}\left(s_k;\btheta_{0,k}\right)\Big|_{s_k = s^{\left(0\right)}}\times 
	\ba^T\bR\left(\bt_k\right)\ba \overset{\left(\textrm{b}\right)}{>}0,
\end{align}
where 
$\left(\textrm{a}\right)$ comes from that from \eqref{equ:eta 0,k,l} we have $\eta_{0,k,\ell}\left(\btheta_{0,k}\right) = p_{0,k}\left(s_k;\btheta_{0,k}\right)\Big|_{s_k = s^{\left(\ell\right)}}, \ell \in \setposi{L-1}$,  and $\sum_{\ell=0}^{L-1}  p_{0,k}\left(s_k;\btheta_{0,k}\right)\Big|_{s_k = s^{\left(\ell\right)}} = 1$,
and $\left(\textrm{b}\right)$ comes from that $p_{0,k}\left(s_k;\btheta_{0,k}\right)\Big|_{s_k = s^{\left(0\right)}} > 0$, and $\bR\left(\bt_k\right)$ is positive definite.
Hence,  $\boldsymbol{\mathcal{I}}_{0,k}\left(\btheta_{0,k}\right), k\in \setposi{2K}$, is positive definite.
From \eqref{equ:FIM of p02}, we can readily obtain that $\boldsymbol{\mathcal{I}}_0\left(\btheta_{0}\right)$ is also positive definite.
Also, we can obtain that $\psi_0\left(\btheta_{0}\right)$ is a strictly convex function of $\btheta_{0}$.

We now show that the problem in \eqref{equ:mp in detection 1} has a minimum and the solution of it is unique, which satisfies \eqref{equ:mp invariant1}.
From \eqref{equ:K-L divergence in detection}, the partial derivative of $\Dkl{p\left(\bs\right)}{p_0\left(\bs;\btheta_0\right)}$ is 
\begin{align}
	\frac{\partial D_{\textrm{KL}}}{\partial \btheta_{0}} &= -\sum_{\bs \in \mathbb{S}^K}p\left(\bs\right)\frac{\partial \ln p_0\left(\bs;\btheta_0\right)}{\partial \btheta_{0}}\nonumber \\
	&=-\sum_{\bs \in \mathbb{S}^K}p\left(\bs\right)\frac{\partial\left( \btheta_0^T\bt - \psi_0\left(\btheta_{0}\right) \right)}{\partial \btheta_{0}}\nonumber\\
	&= - \sum_{\bs \in \mathbb{S}^K}p\left(\bs\right)\left( \bt - \bbeta_0\left(\btheta_{0}\right) \right)
	= - \bbeta + \bbeta_0\left(\btheta_{0}\right).
\end{align}
The Hessian of $\Dkl{p\left(\bs\right)}{p_0\left(\bs;\btheta_0\right)}$ is 
\begin{align}
	\frac{\partial^2 D_{\textrm{KL}}}{\partial \btheta_{0}\partial \btheta_{0}^T} = \frac{\partial \left( - \bbeta + \bbeta_0\left(\btheta_{0}\right) \right) }{\partial \btheta_{0}^T} = \boldsymbol{\mathcal{I}}_0\left(\btheta_{0}\right).
\end{align}
Since the FIM $\boldsymbol{\mathcal{I}}_0\left(\btheta_{0}\right)$ is positive definite, $\Dkl{p\left(\bs\right)}{p_0\left(\bs;\btheta_0\right)}$ is a strictly convex function of $\btheta_{0}$ and \eqref{equ:mp in detection 1} has a unique solution $\btheta_{0}^{\star}$, which satisfies the 
the first order sufficient condition, i.e.,
\begin{equation}
 \bbeta_0\left(\btheta_{0}^{\star}\right) = \bbeta.
\end{equation}
This completes the proof.

\section{Proof of Corollary \ref{the:equivalent of mp2}}\label{proof:equivalent of mp2}
Denote $\bbeta$ as $\bbeta \triangleq \left[\bbeta_{1}^T,\bbeta_{2}^T,\ldots,\bbeta_{2K}^T \right]^T \in \bbR^{2K\left(L-1\right)}$, where $\bbeta_{k} \triangleq \Exp_{p\left(\bs\right)}\braces{\bt_k}\in \bbR^{\left(L-1\right)}$.
Denote the $\ell$-th element in $\bbeta_{k}$ as $\eta_{k,\ell}$, where $k\in \setposi{2K}$ and $\ell \in \setposi{L-1}$.
Then, from $\bt_k = \left[t_{k,1},t_{k,2},\ldots,t_{k,L-1}  \right]^T$ and $t_{k,\ell} = \delta\left( s_k - s^{\left(\ell\right)} \right)$, we can obtain
\begin{align}\label{equ:mariginal}
	\eta_{k,\ell} &= \Exp_{p\left(\bs\right)}\braces{\delta\left( s_k - s^{\left(\ell\right)} \right)} 
	= \sum_{\bs \in \mathbb{S}^{2K}}p\left(\bs\right)\delta\left( s_k - s^{\left(\ell\right)} \right) \nonumber \\
	&= \sum_{\bs_{\setminus k}\in \mathbb{S}^{2K-1}}p\left( \bs \right)\Big|_{s_k = s^{\left(\ell\right)}} 
	\equaa p_k\left(s_k\right)\Big|_{s_k = s^{\left(\ell\right)}},\nonumber
\end{align}
where $ k\in \setposi{2K}$, $\ell \in \setposi{L-1}$, 
$\left(\textrm{a}\right)$ comes from \eqref{equ:marginals of p(s)},
and $p_k\left(s_k\right)$ is the marginal distribution of $s_k$ given $p\left(\bs\right)$.
Similar with the process in Appendix \ref{proof:equivalent of mp1}, $\bbeta_0\left(\btheta_{0}^{\star}\right)$ can be expressed as
$\bbeta_0\left(\btheta_{0}^{\star}\right) \triangleq \left[ \bbeta_{0,1}^T\left(\btheta_{0,1}^{\star}\right),\bbeta_{0,2}^T\left(\btheta_{0,2}^{\star}\right),\ldots,\bbeta_{0,2K}^T\left(\btheta_{0,2K}^{\star} \right) \right]^T \in \bbR^{2K\left(L-1\right)}$,
where $\bbeta_{0,k}\left( \btheta_{0}^{\star} \right) \triangleq \Exp_{p_{0,k}\left(s_k;\btheta_{0,k}^{\star}\right)}\braces{\bt_k}\in \bbR^{\left(L-1\right)}$.
Denote the $\ell$-th element in $\bbeta_{0,k}\left( \btheta_{0}^{\star} \right)$ as $\eta_{0,k,\ell}\left(\btheta_{0}^{\star}\right)$.
We can obtain 
\begin{equation}
   \eta_{0,k,\ell}\left(\btheta_{0}^{\star}\right)	= p_{0,k}\left(s_k;\btheta_{0,k}^{\star}\right)\Big|_{s_k = s^{\left(\ell\right)}},
\end{equation}
through a process the same as that of \eqref{equ:eta 0,k,l}, where $p_{0,k}\left(s_k;\btheta_{0,k}^{\star}\right)$ is the (marginal) probability distribution of $s_k$ given the joint probability distribution $p_0\left(\bs;\btheta_{0}^{\star}\right)$.
Thus, $\bbeta = \bbeta_0\left(\btheta_{0}^{\star}\right)$ is equivalent to 
\begin{equation}\label{equ:equation of marginals2}
			p_k\left(s_k\right) = p_{0,k}\left(s_k;\btheta_{0,k}^{\star} \right), s_k\in \mathbb{S}, k\in \setposi{2K}.
\end{equation}
From Theorem \ref{lemma:equivalent of mp1}, 
$\bbeta = \bbeta_0\left(\btheta_{0}^{\star}\right)$ is a necessary and sufficient condition for $p_0\left(\bs;\btheta_{0}^{\star}\right)$ being the $m$-projection of $p\left(\bs\right)$ onto $\mathcal{M}_0$.
Thus, \eqref{equ:equation of marginals2} is also a necessary and sufficient condition for $p_0\left(\bs;\btheta_{0}^{\star}\right)$ being the $m$-projection of $p\left(\bs\right)$ onto $\mathcal{M}_0$.
Then, combining \eqref{equ:relation} and \eqref{equ:mp3}, we can immediately obtain \eqref{equ:mp solution 0}
This completes the proof.

\section{Proof of Theorem \ref{the:CLT of Y_{n,k}}}\label{proof:CLT of Y_{n,k}}
From the last equation of \eqref{equ:Yn} we have
\begin{align}\label{equ:Y_{n,k} in App}
	Y_{n,k} = \sum_{k'=1, k'\neq k}^{2K}g_{n,k'}s_{k'} + w'_{n,k}
\end{align}
Given $n$ and $k$, let $\braces{X_{k'}}_{k'=1}^{2K}$ be a sequence of random variables, of which each element is defined as
\begin{equation}
	\begin{cases}
		X_{k'} = w'_{n,k},    &k' = k , \\
		X_{k'} =  g_{n,k'}s_{k'},   &\textrm{otherwise}.
	\end{cases}
\end{equation}
Then, $Y_{n,k}$ is the sum of the sequence $\braces{X_{k'}}_{k'=1}^{2K}$.
From the probability distribution of $s_{k'}$ and $w'_{n,k}$ in \eqref{equ:Yn} (also in \eqref{equ:Y_{n,k} in App}) we have
\begin{subequations}
	\begin{align}
		\Exp\braces{X_{k'}} = 
		\begin{cases}
			g_{n,k}s_{k}, & k' = k,\\
			g_{n,k'}\mu_{n,k'}, &\textrm{otherwise},
		\end{cases}
	\end{align}
	\begin{align}\label{equ:variance of X_{k'}}
				\mathbb{V}\braces{X_{k'}} = 
		\begin{cases}
			\sigma_z^2, & k' = k,\\
			g_{n,k'}^2v_{n,k'}, &\textrm{otherwise}.
		\end{cases}
	\end{align}
\end{subequations}
Next, we show that when \eqref{equ:condition on v_{n,k}} holds, the sequence $\braces{X_{k'}}_{k'=1}^{2K}$ satisfies the Lyapunov's condition \eqref{equ:Lyapunov condition}, where $\delta = 1$.
When $k'=k$, $X_{k'}$ is a real Gaussian random variable, and its third central absolute moment is given by \cite{DBLP:journals/corr/Winkelbauer12}
\begin{equation}
	\Exp\braces{\left|X_{k'}  - \Exp\braces{X_{k'}} \right|^3} = 2\sigma_z^3\sqrt{\frac{2}{\pi}} = 2\sigma_z\sqrt{\frac{2}{\pi}}\mathbb{V}\braces{X_{k'}},
\end{equation}
which is bounded.
When $k'\neq k$, we have $X_{k'} = g_{n,k'}s_{k'}$.
Since $s_{k'} \in \mathbb{S}$ and $\mu_{n,k'} = \Exp\braces{s_{k'}}$ are bounded, we can obtain that $X_{k'}$, $\Exp\braces{X_{k'}}$ and $\left( X_{k'}- \Exp\braces{X_{k'}}  \right)$ are also bounded when $g_{n,k'}$ is bounded.
Suppose that $\left|X_{k'} - \Exp\braces{X_{k'}} \right| \le \epsilon$, we can readily obtain that
\begin{equation}
	\Exp\braces{\left| X_{k'} - \Exp\braces{X_{k'}}  \right|^3} \le \epsilon\mathbb{V}\braces{X_{k'}}.
\end{equation}
Let $\varepsilon = \max\left( \epsilon,2\sigma_z\sqrt{2/\pi} \right)$, then for $k'\in \setposi{2K}$ we can obtain
\begin{equation}
	\Exp\braces{\left| X_{k'} - \Exp\braces{X_{k'}}  \right|^3} \le \varepsilon \mathbb{V}\braces{X_{k'}}.
\end{equation}  
Let $\delta = 1$, the Lyapunov's condition for $Y_{n,k} = \sum_{k'=1}^{2K}X_{k'}$ can be expressed as
\begin{align}
0\le	& \ \frac{1}{\left( {\mathbb{V}\braces{Y_{n,k}}}\right)^{\frac{3}{2}} }\sum_{k'=1}^{2K}\Exp\braces{\left| X_{k'} - \Exp\braces{X_{k'}}  \right|^3}\nonumber\\
	\le& \  \frac{\varepsilon }{\left( {\mathbb{V}\braces{Y_{n,k}}}\right)^{\frac{3}{2}} }\sum_{k'=1}^{2K}\mathbb{V}\braces{X_{k'}} \equaa  \frac{\varepsilon}{\sqrt{\mathbb{V}\braces{Y_{n,k}}}},
\end{align}
where $\left(\textrm{a}\right)$ comes from \eqref{equ:variance of Y_{n,k}} and \eqref{equ:variance of X_{k'}}.
Meanwhile, from \eqref{equ:variance of Y_{n,k}} and \eqref{equ:condition on v_{n,k}} we can obtain
\begin{equation}\label{equ:limit of variance of Y_{n,k}}
	\lim\limits_{K\to\infty}\mathbb{V}\braces{Y_{n,k}} = \lim\limits_{K\to\infty}2\zeta K + \sigma_z^2 \to \infty.
\end{equation}
Thus,
\begin{equation}
	\lim\limits_{2K\to\infty}\frac{1}{\left( {\mathbb{V}\braces{Y_{n,k}}}\right)^{\frac{3}{2}} }\sum_{k'=1}^{2K}\Exp\braces{\left| X_{k'} - \Exp\braces{X_{k'}}  \right|^3} = 0.
\end{equation}
Hence, $\braces{X_{k'}}_{k'=1}^{2K}$ satisfies the Lyapunov's condition.
This completes the proof.

\bibliographystyle{IEEEtran}  
\bibliography{IEEEabrv,reference}

\end{document}